\documentclass[aps,pra,showpacs,amsmath,amssymb,superscriptaddress,reprint,10pt,longbibliography]{revtex4-1}
\usepackage{graphicx,mathrsfs,enumerate}
\usepackage{mathtools}
\usepackage{xcolor}
\usepackage{url}
\usepackage{times}%{amssymb,amsthm,amsfonts,amstext}
\usepackage{bbm}
\usepackage{amssymb,amsthm,amsfonts,amstext,color}
\usepackage[english]{babel}
\usepackage[colorlinks=true,breaklinks=true]{hyperref}
\usepackage{comment}

\usepackage{paralist} %for compactenum

\definecolor{jens}{rgb}{0,0.8,0.4}

\newcommand\id{{\mathbbm{1}}}

\newcommand{\bra}[1]{\langle #1|}
\newcommand{\ket}[1]{|#1\rangle}
\newcommand{\braket}[2]{\langle #1|#2\rangle}
\newcommand{\ketbra}[2]{| #1 \rangle \langle #2 |}

\newcommand{\proj}[1]{\vert #1\rangle\!\langle#1 \vert}

\DeclareMathOperator{\Tr}{Tr}

\newtheorem{thm}{Theorem}
\newtheorem{dfn}{Definition}

\newtheorem{lem}{Lemma}
\begin{document}

\title{The resource theory of steering}%Steering as a resource: operational framework, quantification, and nonexistence of steering bits}  

\author{Rodrigo Gallego}
\affiliation{Dahlem Center for Complex Quantum Systems, Freie Universit\"at Berlin, 14195 Berlin, Germany}
\author{Leandro Aolita}
\affiliation{Dahlem Center for Complex Quantum Systems, Freie Universit\"at Berlin, 14195 Berlin, Germany}
\affiliation{Instituto de F\'isica, Universidade Federal do Rio de Janeiro, P. O. Box 68528, Rio de Janeiro, RJ 21941-972, Brazil}

%%%%%%%%%%%%%%%%%%%%%%%%%%%%%%%%%%%%%%%%%%%%%%%%%%%%%%%%%%%%%%%%%%%%%%%%%%%%%%%%%
%%%%%%%%%%%%%%%%%%%%%%%%%%%%%%%%%%%%%%%%%%%%%%%%%%%%%%%%%%%%%%%%%%%%%%%%%%%%%%%%%
\begin{abstract}
We present an operational framework for Einstein-Podolsky-Rosen steering  as a physical resource. For arbitrary-dimensional bipartite systems composed of a quantum subsystem and a black-box device, we show that local operations assisted by one-way classical communication (1W-LOCCs) from the quantum part to the black box cannot create steering. Based on this, we build a resource theory of steering with 1W-LOCCs  as the \emph{free operations}. We introduce the notion of \emph{convex steering monotones} as the fundamental axiomatic quantifiers of steering. As a convenient example thereof, we present the \emph{relative entropy of steering}. In addition, we  prove that two previously proposed quantifiers, the steerable weight and the robustness of steering, are also convex steering monotones. To end up with, for minimal-dimensional systems, we establish, on the one hand, necessary and sufficient conditions for pure-state steering conversions under stochastic 1W-LOCCs and prove, on the other hand, the non-existence of \emph{steering bits}, i.e., measure-independent maximally steerable states from which all states can be obtained by means of the free operations. Our findings reveal unexpected aspects of steering and lay foundations for further research, with potential implications in Bell non-locality. 
\end{abstract}
%%%%%%%%%%%%%%%%%%%%%%%%%%%%%%%%%%%%%%%%%%%%%%%%%%%%%%%%%%%%%%%%%%%%%%%%%%%%%%%%%
%%%%%%%%%%%%%%%%%%%%%%%%%%%%%%%%%%%%%%%%%%%%%%%%%%%%%%%%%%%%%%%%%%%%%%%%%%%%%%%%%

%\pacs{05.70.Ln, 05.20.-y, 05.30.-d, 05.50.+q} 

\maketitle 
%%%%%%%%%%%%%%%%%%%%%%%%%%%%%%%%%%%%%%%%%%%%%%%%%%%%%%%%%%%%%%%%%%%%%%%%%%%%%%%%%
%%%%%%%%%%%%%%%%%%%%% Intro %%%%%%%%%%%%%%%%%%%%%%%%%%%%%%%%%%%%%%%%%%%%%%%%%%%
%\section{Introduction}
\section{Introduction}
Steering, as Schr\"odinger named it \cite{Schrodinger35}, is an exotic quantum effect by which ensembles of quantum states can be remotely prepared by performing local measurements at a distant lab. 
It allows \cite{Wiseman07, Reid09} to certify the presence of entanglement between a user with an untrusted measurement apparatus,  Alice, and another with a trusted quantum-measurement device, Bob.
 Thus, it constitutes a fundamental notion between quantum entanglement \cite{Horodecki09}, whose certification requires quantum measurements on both sides, and Bell non-locality \cite{Brunner13}, where both users possess untrusted black-box devices. 
Steering can be detected through simple tests analogous to Bell inequalities \cite{Cavalcanti09}, and has been verified in a variety of remarkable experiments \cite{Experiments}, including steering without Bell non-locality \cite{Saunders10} and a fully loop-hole free steering demonstration \cite{loopholefree}.  Apart from its fundamental relevance, steering has been identified as a resource for one-sided (1S) device-independent (DI) quantum key-distribution (QKD), where only one of the parts has an untrusted apparatus while the other ones possess trusted devices \cite{Branciard12, He13}. 
There, the experimental requirements for unconditionally secure keys are less stringent than in fully (both-sided) DI-QKD \cite{Deviceindependent}. 

The formal treatment of a physical property as a resource is given by a \emph{resource theory}. The basic component of this is a restricted class of operations, called the \emph{free operations}. These are typically the operations at hand in physical scenarios where the property in question acts as a useful resource. They fulfil the essential requirement of mapping every \emph{free state}, i.e., every one without the property, into a free state. %In this way, the resourceful states can be defined as those not attainable by free operations acting on free states.
Furthermore, they provide a formal recipe for the quantification of the resource:
The fundamental necessary condition for a function to be a measure of the resource is that it is monotonous --non-increasing-- under the free operations. 
That is, the operations that do not increase the resource on the free states do not increase it on all other states either. 

Entanglement theory is the most popular and best understood  \cite{VedralPlenio,Brandao08} resource theory. There, local operations assisted by classical communication (LOCCs) are usually the natural free operations \cite{Bennett96}. However, other sets of operations that do not create entanglement either have also been  considered \cite{Jonathan99,Vedral97}.
On the other hand, resource theories have been formulated also for states out of thermal equilibrium \cite{Brandao13}, asymmetry \cite{Ahmadi13}, reference frames \cite{Gour08}, and non-locality \cite{Gallego12,deVicente14}, for instance.

In steering theory, systems are described by a collection of ensembles of quantum states on Bob's side and a conditional probability distribution of measurement outcomes (outputs) given measurement settings (inputs) on Alice's. Each input of Alice's is associated with an ensemble of Bob's and each output with a member state of such ensemble. That is, each pair of inputs and outputs of Alice's is correlated with a state in an ensemble of Bob's. Such systems are called {\it assemblages} \cite{Pusey13,Skrzypczyk13,Bowles14,Quintino14, Piani15,Belen15}. The free operations for steering
must thus arise from natural constrains native of a physical scenario where steerable assemblages are useful for some task. Up to now, no attempt for an operational framework of steering as a resource had been reported.

In this work we develop the resource theory of steering. First, we observe that one-way LOCCs (1W-LOCCs) from Bob to Alice are allowed operations --in the sense of not compromising the security-- in 1S-DI-QKD protocols, the physical scenario where steering is a known resource \cite{Branciard12, He13}. Then, for arbitrarily many inputs and outputs for Alice's black box and arbitrary Hilbert-space dimension for Bob's quantum system, we show that 1W-LOCCs from Bob to Alice do not create steering. These two facts give us clear physical motivations to take 1W-LOCCs as natural free operations for steering. We present the explicit parametrisation of a generic 1W-LOCC acting on assemblages. With this, we provide a formal definition of steering monotones.
As an example thereof, we present the relative entropy of steering, for which we also introduce, on the way, the notion of relative entropy between assemblages. 
In addition, we prove 1W-LOCC monotonicity for two other recently proposed steering measures, the steerable weight \cite{Skrzypczyk13} and the robustness of steering \cite{Piani15}, and convexity for all three measures.
To end up with, we prove two theorems on steering conversion under stochastic 1W-LOCCs for the lowest-dimensional case, i.e., qubits on Bob's side and 2 inputs $\times$ 2 outputs on Alice's. In the first one, we 
show that it is impossible to transform via 1W-LOCCs, not even probabilistically,  an assemblage composed of pairs of pure orthogonal states into another assemblage composed also of pairs of pure orthogonal states but with a different pair overlap, unless the latter is unsteerable. This yields infinitely many inequivalent classes of steering already for systems of lowest dimension.
In the second one, we show that there exists no assemblage composed of pairs of pure states that can be transformed into any assemblage by stochastic 1W-LOCCs.
This implies, in striking contrast to entanglement theory,  that there exists no operationally well defined, measure-independent maximally steerable assemblage of minimal dimension. 

The paper is organized as follows. In Sec. \ref{sec:scenario} we formally define assemblages and present their basic properties. In Sec. \ref{sec:1WLOCC_1SDIQKD} we discuss the role of 1W-LOCCs as the natural operations available to Alice and Bob in 1S-DI-QKD. In Sec. \ref{sec:theoperationalframework} we find an explicit parametrisation of all stochastic 1W-LOCCs from assemblages into assemblages and show that the resulting maps are steering non-increasing operations. In Sec. \ref{sec:steeringmonotonicity} we introduce the notion of convex steering monotones. In Sec. \ref{sec:therelativeentropyofsteering} we present the relative entropy of steering. In Sec. \ref{sec:other_measures} we show convexity and 1W-LOCC-monotonicity of the steerable weight and the robustness of steering. In Sec. \ref{sec:assemblageconversions} we study, for minimal-dimensional systems, assemblage conversions under 1W-LOCCs and prove the in existence of pure-assemblage steering bits. Finally, in Sec. \ref{sec:discussionandoutlook} we present our conclusions and mention some future research directions that our results offer.
%%%%%%%%%%%%%%%%%%%%%%%%%%%%%%%%%%%%%%%%%%%%%%%%%%%%%%%%%%%%%%%%%%%%%%%%%%%%%%%%%
%%%%%%%%%%%%%%%%%%%%%%%%%%%%%%%%%%%%%%%%%%%%%%%%%%%%%%%%%%%%%%%%%%%%%%%%%%%%%%%%%
%\section{Definitions} 
\section{Assemblages and steering}\label{sec:scenario}
We consider two distant parties, Alice and Bob, who have each a half of a bipartite system. Alice holds a so-called black-box device, which, given a classical input $x\in[s]$, generates a classical output $a\in[r]$, where $s$ and $r$ are natural numbers and the notation $[n]\coloneqq\{0,\ldots,n-1\}$, for $n \in\mathbb{N}$, is introduced. Bob holds a quantum system of dimension $d$ (\emph{qudit}), whose state he can perfectly characterize tomographically via trusted quantum measurements.
The joint state of their system is thus fully specified by an {\it assemblage}
\begin{equation}
\label{eq:ensembledef}
\rho_{A|X}\coloneqq\{P_{A|X}(a,x),\varrho(a,x)\}_{a\in[r],x\in[s]},
\end{equation} 
of normalized quantum states $\varrho(a,x)\in\mathcal{L}(\mathcal{H}_B)$, with $\mathcal{L}(\mathcal{H}_B)$ the set of linear operators on Bob's subsystem's Hilbert space $\mathcal{H}_B$, each one associated to a conditional probability $P_{A|X}(a,x)$ of Alice getting an output $a$ given an input $x$. We denote by $P_{A|X}$ the corresponding conditional probability distribution. 
%Set \eqref{eq:ensembledef} is sometimes called  \cite{Pusey13,Skrzypczyk13}, and can be experimentally obtained by Alice choosing different inputs and Bob measuring in tomographycally complete bases. 

Equivalently, each pair $\{P_{A|X}(a,x),\varrho(a,x)\}$ can be univocally represented by the unnormalized quantum state 
\begin{equation}
\label{eq:unnorm_state}
\varrho_{A|X}(a,x)\coloneqq P_{A|X}(a,x)\times \varrho(a,x).
\end{equation} 
%and the assemblage \eqref{eq:ensembledef} by the map $\varrho_{A|X}$ from input-output pairs $(a,x)$ into unormalized quantum states $\lo{\varrho}_{A|X}(a,x)$. \la{quitar este ultimo comment?}
In turn, an alternative representation of the assemblage $\rho_{A|X}$ is given by the set $\hat{\rho}_{A|X}\coloneqq\{\hat{\rho}_{A|X}(x)\}_x$ of quantum states 
\begin{equation}
\label{eq:quant_rep}
\hat{\rho}_{A|X}(x)\coloneqq\sum_{a} \proj{a} \otimes \varrho_{A|X}(a,x)\in\mathcal{L}(\mathcal{H}_E\otimes\mathcal{H}_B),
\end{equation}
where $\{\ket{a}\}$ is an orthonormal basis of an auxiliary extension Hilbert space $\mathcal{H}_E$ of dimension $r$. The states $\{\ket{a}\}$ do not describe the system inside Alice's box, they are just abstract flag states to represent its outcomes with a convenient bra-ket notation. Expression \eqref{eq:quant_rep} gives the counterpart for assemblages of the so-called extended Hilbert space representation used for ensembles of quantum states \cite{Oreshkov09}. We refer to $\hat{\rho}_{A|X}$
for short as the \emph{quantum representation} of $\rho_{A|X}$ and use either notation upon convenience.

We restrict throughout to \emph{no-signaling assemblages}, i.e., those for which Bob's reduced state $\varrho_B\in\mathcal{L}(\mathcal{H}_B)$ does not depend on Alice's input choice $x$:
\begin{equation}
\label{eq:nosignaling}
\varrho_B\coloneqq\sum_{a}\varrho_{A|X}(a,x)=\sum_{a}\varrho_{A|X}(a,x') \:\:\: \forall\ x,x'.
\end{equation}
The assemblages fulfilling the no-signaling condition \eqref{eq:nosignaling} are the ones that possess a \emph{quantum realization}. That is, they can be obtained from local quantum measurements by Alice on a joint quantum state $\varrho_{AB}\in\mathcal{L}(\mathcal{H}_A\otimes\mathcal{H}_B)$ shared with Bob, where $\mathcal{H}_A$ is the Hilbert space of the system inside Alice's box. For any no-signaling assemblage $\rho_{A|X}$, we refer as the \emph{trace of the assemblage}  to the $x$-independent quantity 
\begin{equation}
\label{eq:trace_ass}
\Tr[\rho_{A|X}]\coloneqq\Tr_{EB}[\hat{\rho}_{A|X}]=\Tr[\varrho_B]=\sum_{a}P_{A|X}(a,x),
\end{equation}
and say that the assemblage is normalized if $\Tr[\rho_{A|X}]=1$ and unnormalized if $\Tr[\rho_{A|X}]\leq1$.

An assemblage $\sigma_{A|X}\coloneqq\{\varsigma_{A|X}(a,x)\}_{a\in[r],x\in[s]}$, being $\varsigma_{A|X}(a,x)\in\mathcal{L}(\mathcal{H}_B)$ unnormalised states, is called \emph{unsteerable} if there exist a probability distribution $P_{\Lambda}$, a  conditional probability distribution $P_{A|X \Lambda}$, and normalized states $\xi(\lambda)\in\mathcal{L}(\mathcal{H}_B)$ such that
\begin{equation}
\label{eq:lhs}
\varsigma_{A|X}(a,x)=\sum_{\lambda}P_{\Lambda}(\lambda)P_{A|X \Lambda}(a,x,\lambda) \:\xi(\lambda)\:\:\: \forall\ x,a.
\end{equation}
Such assemblages %are the ones compatible with classical correlations. They 
can be obtained by sending a shared classical random variable $\lambda$ to Alice, correlated with the state $\xi(\lambda)$ sent to Bob, and letting Alice classically post-process her random variable according to  $P_{A|X \Lambda}$, with $P_{X,\Lambda}=P_X \times P_{\Lambda}$ so that condition \eqref{eq:nosignaling} holds.  The variable $\lambda$ is called a \emph{local-hidden variable} and the decomposition \eqref{eq:lhs} is accordingly referred to as a \emph{local-hidden state} (LHS) model. We refer to the set of all unsteerable assemblages as $\mathsf{LHS}$. Any assemblage that does not admit a LHS model as in Eq. \eqref{eq:lhs} is called \emph{steerable}. An assemblage is compatible with classical correlations if, and only if, it is unsteerable. 

To end up with, a comment on steering as a property of assemblages, as opposed to quantum states, is in order. In the pioneering works~\cite{Wiseman07, Reid09}, steering is defined as a property of quantum states. Namely, there, a state is said to be steerable if it can give rise, under local measurements, to correlations without a LHS model, i.e., to a steerable assemblage. The definition considered here, directly in terms of assemblages, and to which our resource theory applies, follows the treatment of Refs.~\cite{Pusey13,Skrzypczyk13,Bowles14,Quintino14,Piani15,Belen15}, for instance, where one embeds a share of the bipartite quantum state in the untrusted measurement device (Alice's, in our case) and treats the entire embedding as a black box with unknown behaviour. This is the scenario of 1S-DI-QKD, the very task for which steering is a known useful resource. There, quantum states alone are not useful, since, without a trusted measurement device, correlations without LHS model can in general not be obtained from them. Working with assemblages, in contrast, is advantageous precisely because it removes the need of measurement specification in the untrusted part \cite{footnote2}. Finally, it is important to note that post-quantum steering, i.e., steerable assemblages that, in spite of satisfying the no-signalling principle, do not admit a quantum realisation, has been recently discovered \cite{Belen15}.

%%%%%%%%%%%%%%%%%%%%%%%%%%%%%%%%%%%%%%%%%%%%%%%%%%%%%%%%%%%%%%%%%%%%%%%%%%%%%%%%%
%\subsection{The free operations}
%\label{sec:free_ops}
\section{The operational framework}
\label{sec:theoperationalframework}
In this section, we show that stochastic 1W-LOCCs from Bob to Alice do not create steering and can, therefore, be taken as free operations for steering. We consider the general scenario of stochastic 1W-LOCCs, i.e., 1W-LOCCs that do not necessarily occur with certainty, which map the initial assemblage $\rho_{A_f|X_f}$ into a final assemblage $\rho_{A_f|X_f}$. See Fig. \ref{fig:1}.
Bob's generic quantum operation can be represented by an incomplete generalised measurement. This is described by a completely-positive non trace-preserving map $\mathcal{E}:\mathcal{L}(\mathcal{H}_B)\to\mathcal{L}({\mathcal{H}_B}_f)$ defined by 
\begin{subequations}
\label{eq:Kraus_def}
\begin{align}
\mathcal{E}(\cdot)\coloneqq&\sum_{\omega}\mathcal{E}_{\omega}(\cdot),\ \text{with}\ \mathcal{E}_{\omega}(\cdot)\coloneqq K_{\omega} \: \cdot \: K^{\dagger}_{\omega},\\
&\text{such that}\ \sum_{\omega}K^{\dagger}_{\omega}K_{\omega}\leq \id,
\end{align}
\end{subequations}
where ${\mathcal{H}_B}_f$ is the final Hilbert space, of dimension $d_f$, and $K_{\omega}:\mathcal{H}_B\to{\mathcal{H}_B}_f$ is the measurement operator corresponding to the $\omega$-th measurement outcome. For any normalised $\varrho_B\in\mathcal{L}(\mathcal{H}_B)$, the trace $\Tr[\mathcal{E}(\varrho_B)]%=\Tr[(\sum_{\omega}K^{\dagger}_{\omega}K_{\omega})\varrho]
\leq1$ of the map's output $\mathcal{E}(\varrho_B)$ represents the probability that the physical transformation $\varrho_B\to\mathcal{E}(\varrho_B)/\Tr[\mathcal{E}(\varrho_B)]$ takes place. 
In turn, the map $\mathcal{E}_{\omega}(\cdot)$ describes the post-selection of the $\omega$-th outcome, which occurs with a probability 
\begin{equation}
\label{eq:P_Omega}
P_\Omega(\omega)\coloneqq\Tr[\mathcal{E}_{\omega}(\rho_{B})]=\Tr[K_{\omega} \varrho_{B} K^{\dagger}_{\omega}]\leq1.
\end{equation}

%%%%%%%%%%%%%%%%%%%%%%%%%%%%%%%%%%%%%%%%%%%%%%%%%%%%%%%%%%%%%%%%%%%%%%%%%%%%%%%%%
%%%%%%%%%%%%%%%%%%%%%%%%%%%%%%%%%%%%%%%%%%%%%%%%%%%%%%%%%%%%%%%%%%%%%%%%%%%%%%%%%
\begin{figure}[t!]
\centering
\includegraphics[width=1\linewidth]{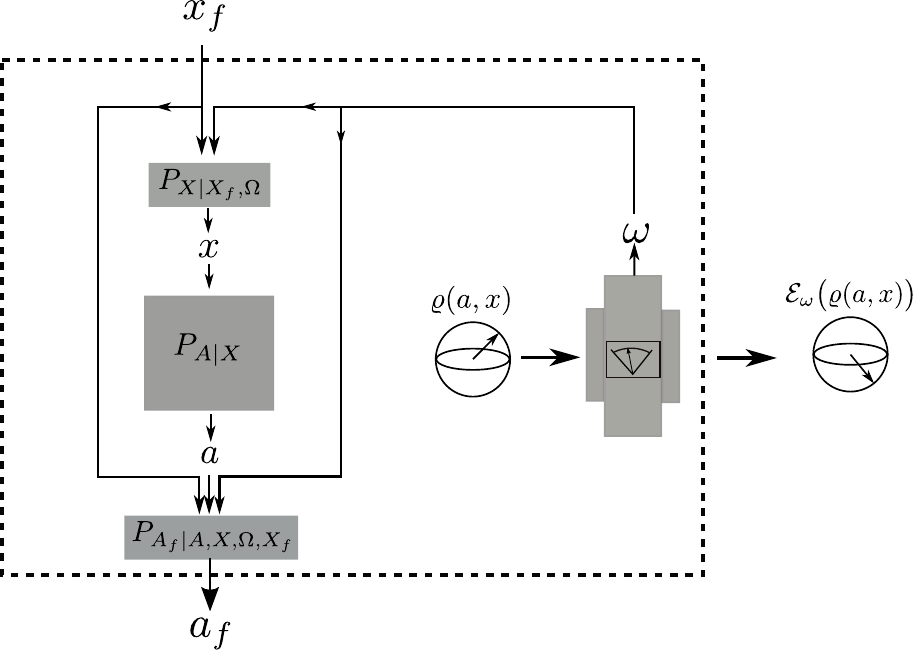}
\caption{Schematic representation of a $\mathsf{1WLOCC}$ map $\mathcal{M}$: The initial assemblage $\rho_{A|X}$ consists of a black-box, with inputs $x$ and outputs $a$, governed by the probability distribution $P_{A|X}$, in Alice's hand, and a quantum subsystem in one of the states $\{\varrho(a,x)\}_{a,x}$, in Bob's hands. The final assemblage $\rho_{A_f|X_f}=\mathcal{M}(\rho_{A|X})$ is given by a final black-box, represented by the dashed-lined rectangle,  of inputs $x_f$ and outputs $a_f$, and a final subsystem, represented outside the dashed-lined rectangle, in the state $\varrho(a_f,x_f)=\mathcal{E}_{\omega}(\varrho(a,x))$. To implement $\mathcal{M}$, first, Bob applies, with a probability $P_\Omega(\omega)$, a stochastic quantum operation $\mathcal{E}_{\omega}$ that leaves his subsystem in the state $\mathcal{E}_{\omega}(\varrho(a,x))$. He communicates $\omega$ to Alice. Then, Alice generates $x$ by processing the classical bits $\omega$ and $x_f$ with a local wiring described by a conditional distribution $P_{X|X_f,\Omega}$. She inputs $x$ to her initial device, upon which the bit $a$ is output. Finally, Alice generates the output $a_f$ of the final device by processing $x_f$, $\omega$, $x$, and $a$, with a local wiring described by a distribution $P_{A_f|A,X,\Omega,X_f}$.}
\label{fig:1}
\end{figure}

Since Alice can only process classical information, the allowed one-way communication from Bob to her must be classical too. Thus, it can only consists of the outcome $\omega$ of his quantum operation.
Classical bit processings are usually referred to as \emph{wirings}  \cite{Brunner13}. Alice's wirings map $a\in[r]$ and $x\in[s]$ into input and out bits $a_f\in[r_f]$ and $x_f\in[s_f]$, respectively, of the final assemblage, where $s_f$ and $r_f$ are natural numbers. 
The most general wirings respecting the above constraints are described by conditional probability distributions $P_{X|X_f,\Omega}$ and $P_{A_f|A,X,\Omega,X_f}$ of generating $x$ from $\omega$ and $x_f$ and $a_f$ from $x_f$, $\omega$, $x$, and $a$, respectively, as sketched in Fig. \ref{fig:1}. Finally, since, as mentioned, her wirings must be deterministic, $P_{X|X_f,\Omega}$ and $P_{A_f|A,X,\Omega,X_f}$ must be normalised probability-preserving distributions. 

All in all, the general form of the resulting maps is parametrised in the following definition (see App. \ref{sec:genericoperations} for details).

\begin{dfn}[Stochastic assemblage 1W-LOCCs]
\label{Def1WLOCCs}
We define the class $\mathsf{1WLOCC}$ of \emph{(stochastic) 1W-LOCCs} as the set of (stochastic) maps $\mathcal{M}$
that take an arbitrary assemblage $\hat{\rho}_{A|X}$ into a  final assemblage $\hat{\rho}_{A_f|X_f}\coloneqq\mathcal{M}(\hat{\rho}_{A|X})$, where
\begin{equation}
\label{eq:final_assemb}
\mathcal{M}(\hat{\rho}_{A|X})\coloneqq\sum_{\omega} (\id \otimes K_\omega) \: \mathcal{W}_{\omega}(\hat{\rho}_{A|X}) \:  (\id \otimes K_{\omega}^{\dagger}),
\end{equation}
being $\mathcal{W}_{\omega}$ a \emph{deterministic w}iring map given by
\begin{align}
\nonumber
[\mathcal{W}_{\omega} (\hat{\rho}_{A|X})](x_f)&:=\sum_{a_f,a,x} P(x|x_f,\omega)P(a_f|a,x,\omega,x_f)  \\
\label{eq:def_W}
&\times\:\: (\ketbra{a_f}{a}\otimes \id  )  \:  \hat{\rho}_{A|X}(x) \:  ( \ketbra{a}{a_f} \otimes \id),
\end{align}
with $P(x|x_f,\omega)$ and $P(a_f|a,x,\omega,x_f)$ short-hand notations for the conditional probabilities $P_{X|X_f,\Omega}(x,x_f,\omega)$ and $P_{A_f|A,X,\Omega,X_f}(a_f,a,x,\omega,x_f)$, respectively.
\end{dfn}

Note that the final assemblage \eqref{eq:final_assemb} is in general not normalised: Introducing 
\begin{equation}
\label{eq:def_M_mu}
\mathcal{M}_{\omega}(\: \cdot \:)\coloneqq(\id \otimes K_\omega) \: \mathcal{W}_{\omega}(\: \cdot \:) \:  (\id \otimes K_{\omega}^{\dagger}),
\end{equation}
such that $\mathcal{M}(\: \cdot \:)=\sum_{\omega}\mathcal{M}_{\omega}(\: \cdot \:)$, we obtain, using Eqs. \eqref{eq:quant_rep},  \eqref{eq:nosignaling}, \eqref{eq:trace_ass}, \eqref{eq:P_Omega}, \eqref{eq:final_assemb}, and \eqref{eq:def_W}, that
\begin{equation}
\label{eq:prob_tranform}
\Tr[\mathcal{M}(\hat{\rho}_{A|X})]=\sum_{\omega}\Tr[\mathcal{M}_{\omega}(\hat{\rho}_{A|X})]=\sum_{\omega}P_{\Omega}(\omega) \leq 1.
\end{equation}
As with quantum operations, the trace \eqref{eq:prob_tranform} of $\mathcal{M}(\hat{\rho}_{A|X})$ represents the probability that the physical transformation $\hat{\rho}_{A|X}\to\mathcal{M}(\hat{\rho}_{A|X})/\Tr[\mathcal{M}(\hat{\rho}_{A|X})]$ takes place. Analogously, 
the  map $\mathcal{M}_{\omega}$ describes the assemblage transformation that takes place when Bob post-selects the $\omega$-th outcome, which occurs with probability $\Tr[\mathcal{M}_{\omega}(\hat{\rho}_{A|X})]=P_{\Omega}(\omega)$. In the particular case where $\mathcal{M}$ is trace-preserving, we refer to it as a \emph{deterministic 1W-LOCC}. 

Finally, we prove in App. \ref{sec:proof_Theorem1WLOCCs} the following theorem.
\begin{thm}[$\mathsf{1WLOCC}$ invariance of $\mathsf{LHS}$]
\label{Theorem1WLOCCs}
Any map of the class $\mathsf{1WLOCC}$ takes every unsteerable assemblage into an unsteerable assemblage.
\end{thm}

%%%%%%%%%%%%%%%%%%%%%%%%%%%%%%%%%%%%%%%%%%%%%%%%%%%%%%%%%%%%%%%%%%%%%%%%%%%%%%%%%
%%%%%%%%%%%%%%%%%%%%%%%%%%%%%%%%%%%%%%%%%%%%%%%%%%%%%%%%%%%%%%%%%%%%%%%%%%%%%%%%%
%\section{Free operations for steering}
\section{Physical motivation for free operations: 1W-LOCCs as safe operations in 1S-DI-QKD}
\label{sec:1WLOCC_1SDIQKD}
%%%%%%%%%%%%%%%%%%%%%%%%%%%%%%%%%%%%%%%%%%%%%%%%%%%%%%%%%%%%%%%%%%%%%%%%%%%%%%%%%
As shown in the previous section, stochastic 1W-LOCCs from Bob to Alice satisfy the basic requirement of mapping every unsteerable assemblage into an unsteerable assemblage. However, there may in general exist other sets of operations with this feature. In entanglement theory, for instance, apart from the LOCCs, the separable operations \cite{Vedral97}, the entanglement-assisted (catalytic) LOCCs \cite{Jonathan99}, or simply the local operations, as well as any of these supplemented with particle swapping \cite{Horodecki09}, are known not to create entanglement either. Each of the these classes of operations leads, strictly speaking, to a valid resource theory of entanglement. The choice of a given class over others is based upon actual constraints from the physical scenario in question. In what follows, we discuss the role of 1W-LOCCs as the allowed operations in 1S-DI-QKD, in the sense of being those that do not compromise the security. This gives us a physical motivation to choose 1W-LOCC as free operations for steering over other classes of operations that may also map $\mathsf{LHS}$ into itself. To this end, and for pedagogic reasons, we first discuss the allowed safe operations for QKD and fully DI QKD.

QKD consists of the extraction of a secret key from the correlations of local-measurement outcomes on a bipartite quantum state. The most fundamental constraint to which any generic QKD protocol is subject is, of course, the lack of a  private safe classical-communication channel between distant labs. If such a channel were available, the whole enterprise of QKD would be pointless. This imposes restrictions on the operations allowed, so as not to break the security of the protocol. For instance, clearly, the local-measurement outcomes cannot be communicated, as they can be intercepted by potential eavesdroppers who could, with them, crack the key.
Of particular relevance for this work are the assumptions on the measurement devices. In non-DI QKD protocols, entanglement is the resource and security is proven under the assumption that the users have a specific quantum state and perfectly characterised measurement devices \cite{Ecker91}. Knowledge of the state by an eavesdropper does not compromise the security. Therefore, prior to the measurements producing the key, the users are allowed to preprocess the state in any way and exchange information about it, for instance with LOCCs, or even to discard the state aborting the protocol run. Pre-processing LOCCs or abortions can at most provide an eavesdropper with knowledge about the state, not about the key, and therefore do not affect the security.

The situation is different in DI-QKD \cite{Deviceindependent}. There, the resource is given by Bell non-local correlations and no assumption is made either on the quantum state or the measurement devices. The users effectively hold black-box devices, whose inputs and outputs are all to which they have access. Since such inputs and outputs are precisely the bits with which the key is established, both classical communication and abortions are forbidden. Communication of outputs can directly reveal the key, as mentioned, whereas abortions and communication of inputs can, due to the locality and detection loopholes, respectively, be maliciously exploited by an eavesdropper to obtain information about the key too. 
Hence, the security constrains of DI-QKD naturally yield local classical information processing assisted by shared randomness or prior-to-input classical communication as the class of allowed operations \cite{Gallego12,deVicente14}. 

In 1S-DIQKD, in contrast, while no assumption is made on the bipartite quantum state or Alice's apparatus, Bob's measurement device is perfectly characterised.
This is effectively described by assemblages of the form given in Eq. \eqref{eq:ensembledef}. 
The asymmetry in Alice and Bob's devices leads to an asymmetry in the operations allowed to each of them. Alice is subject to the same restrictions as in both-sided DI QKD, while Bob, to those of non-DI QKD. Hence, Alice cannot abort or transmit any information, but, before measuring, Bob is allowed to implement arbitrary pre-processing quantum operations to his subsystem, including stochastic ones with possible abortions, and send any classical feedback about them to Alice. Altogether, this singles out a natural set of operations that do not compromise the security: all the assemblage transformations involving only deterministic classical maps on Alice's side and arbitrary --possibly stochastic-- quantum operations on Bob's, assisted by one-way classical communication from Bob to Alice \cite{footnote1}. These are, namely, the stochastic 1W-LOCCs from Bob to Alice  (see Fig. \ref{fig:1}). Note that shared randomness or prior-to-input classical communication from Alice to Bob \cite{Gallego12,deVicente14}, which also do not introduce any security compromise, can always be recast as 1W classical communication from Bob to Alice and need, therefore, not be explicitly considered. 

Finally, we emphasise that it is one of the measurement apparatuses what is assumed untrusted, not the users. Both users are trusted and reliably operate on their systems, carrying out the allowed assemblage transformations.

%%%%%%%%%%%%%%%%%%%%%%%%%%%%%%%%%%%%%%%%%%%%%%%%%%%%%%%%%%%%%%%%%%%%%%%%%%%%%%%%%%%%%%%%%%%%%%%%%%%%%%%%%%%%%%
%%%%%%%%%%%%%%%%%%%%%%%%%%%%%%%%%%%%%%%%%%%%%%%%%%%%%%%%%%%%%%%%%%%%%%%%%%%%%%%%%%%%%%%%%%%%%%%%%%%%%%%%%%%%%%
\section{Steering monotonicity}\label{sec:steeringmonotonicity}
Once the free operations for steering are established, the natural next step is to introduce an axiomatic approach to define steering measures, i.e., a set of postulates that a bona fide quantifier of steering should fulfil. 
%Given the formalism that introduces steering in the framework of a resource theory, it is a natural next step to define functions that allow one to quantify the steering of a given assemblage. In the same way that it is done for quantum states and entanglement, we formulate an axiomatic approach where steering quantifiers should not increase under the allowed operations, in this case S1W-LOCCs.
\begin{dfn}[1W-LOCC-monotonicity and convexity]
\label{def:Mon_Conv}
A function $\mathscr{S}$, from the space of assemblages into $\mathbb{R}_{\geq0}$, is a \emph{steering monotone} if it fulfils the following two axioms:
\begin{enumerate}[i)]
\item $\mathscr{S}(\hat{\rho}_{A|X})=0$ for all $\hat{\rho}_{A|X}\in \mathsf{LHS}$.
\item $\mathscr{S}$ does not increase, on average, under deterministic 1W-LOCCs, i.e.,
\begin{equation}
\label{eq:def_strong_mon}
\sum_{\omega} P_{\Omega}(\omega) \mathscr{S}\left( \frac{\mathcal{M}_{\omega}(\hat{\rho}_{A|X})}{\Tr\left[\mathcal{M}_{\omega}(\hat{\rho}_{A|X})\right]}\right)\leq \mathscr{S}(\hat{\rho}_{A|X})
\end{equation}
for all $\hat{\rho}_{A|X}$, with $P_{\Omega}(\omega)=\Tr\left[\mathcal{M}_{\omega}(\hat{\rho}_{A|X})\right]$ and $\sum_{\omega} P_{\Omega}=1$.
\end{enumerate}
Besides, $\mathscr{S}$ is a \emph{convex steering monotone} if it additionally satisfies the  property:
\begin{enumerate}[i)]
  \setcounter{enumi}{2}
\item Given any real number $0\leq\mu\leq 1$, and assemblages $\hat{\rho}_{A|X}$ and $\hat{\rho}'_{A|X}$, then
\begin{align}
\nonumber
\mathscr{S}\left(\mu\, \hat{\rho}_{A|X}+(1-\mu)\hat{\rho}'_{A|X}\right)&\leq \mu\, \mathscr{S}\left(\hat{\rho}_{A|X}\right)\\
\label{def:convexity}
&+(1-\mu)\mathscr{S}\left(\hat{\rho}'_{A|X}\right).
\end{align}
\end{enumerate}
\end{dfn}
Condition $i)$ reflects the basic fact that unsteerable assemblages should have zero steering. Condition $ii)$ formalizes the intuition that, analogously to entanglement, steering should not increase --on average-- under 1W-LOCCs, even if the flag information $\omega$ produced in the transformation is available.
Finally, condition $iii)$ states the desired property that steering should not increase by probabilistically mixing assemblages. The first two conditions are taken as mandatory necessary conditions, the third one  only as  a convenient property. Importantly, there exists a less demanding definition of monotonicity. There, the left-hand side of Eq. \eqref{eq:def_strong_mon} is replaced by $\mathscr{S}\left(\mathcal{M}(\hat{\rho}_{A|X})/\Tr[\mathcal{M}(\hat{\rho}_{A|X})]\right)$. That is, $ii')$ it is demanded only that steering itself, instead of its average over $\omega$, is non-increasing under the free operations. 
The latter is actually the most fundamental necessary condition for a measure. However, monotonicity $ii)$ is in many cases (including the present work) easier to prove and, together with condition $iii)$, implies monotonicity $ii')$. Hence, we focus throughout on monotonicity as defined by Eq.  \eqref{eq:def_strong_mon} and refer to it simply as \emph{1W-LOCC monotonicity}. All three known quantifiers of steering, the two ones introduced in Refs. \cite{Skrzypczyk13,Piani15} as well as the one we introduce in the next section, turn out to be convex steering monotones in the sense of Definition \ref{def:Mon_Conv}.

%%%%%%%%%%%%%%%%%%%%%%%%%%%%%%%%%%%%%%%%%%%%%%%%%%%%%%%%%%%%%%%%%%%%%%%%%%%%%%%%%%%%%
%%%%%%%%%%%%%%%%%%%%%%%%%%%%%%%%%%%%%%%%%%%%%%%%%%%%%%%%%%%%%%%%%%%%%%%%%%%%%%%%%%%%%
\section{The relative entropy of steering}\label{sec:therelativeentropyofsteering}
In this section, we introduce a convex steering monotone called the \emph{relative entropy of steering}. To this end, we first define the notion of \emph{relative entropy} between assemblages. For any two density operators $\varrho$ and $\varrho'$, we first recall the \emph{quantum von-Neumann relative entropy}
\begin{equation}
\label{eq:von_Neumman_rel_ent}
S_\text{Q}(\varrho\| \varrho')\coloneqq\Tr \left[\varrho\left(\log \varrho -\log \varrho'\right)\right]
\end{equation}
 of $\varrho$ with respect to $\varrho'$ and, for any two probability distributions $P_X$ and $P'_X$, the \emph{classical relative entropy}, or \emph{Kullback-Leibler divergence},
 \begin{equation}
\label{eq:Kullback_Leibler_div}
S_\text{C}(P_X\| P'_X)\coloneqq \sum_x P_X(x) [\log P_X(x)- \log P'_X(x)]
\end{equation}
 of $P_X$ with respect to $P'_X$. 
The quantum and classical relative entropies \eqref{eq:von_Neumman_rel_ent} and \eqref{eq:Kullback_Leibler_div} measure  the distinguishability of states and distributions, respectively. 
To find an equivalent measure for assemblages, we note, for $\hat{\rho}_{A|X}(x)$ given by Eq. \eqref{eq:quant_rep} and $\hat{\rho}'_{A|X}(x)\coloneqq\sum_{a}  P'_{A|X}(a,x)\proj{a} \otimes \varrho'(a,x)$, that 
\begin{align}
\label{eq:towards_rel_ent_assemb}
\nonumber
S_\text{Q}\left(\hat{\rho}_{A|X}(x)\| \hat{\rho}'_{A|X}(x)\right)=S_\text{C}\left(P_{A|X}(\cdot,x)\| P'_{A|X}(\cdot,x)\right)&\\
+\sum_{a}P_{A|X}(a,x)S_\text{Q}\left(\varrho(a,x)\|\varrho'(a,x)\right),&
\end{align}
where $P_{A|X}(\cdot,x)$ and $P'_{A|X}(\cdot,x)$ are respectively the distributions over $a$ obtained from the conditional distributions $P_{A|X}$ and $P'_{A|X}$ for a fixed $x$.
That is, the distinguishability between the states $\hat{\rho}_{A|X}(x)$ and $\hat{\rho}'_{A|X}(x)\in\mathcal{L}(\mathcal{H}_E\otimes\mathcal{H}_B)$ equals the sum of the distinguishabilities between  $P_{A|X}(x)$ and $P'_{A|X}(x)$ and between $\varrho(a,x)$ and $\varrho'(a,x)\in\mathcal{L}(\mathcal{H}_B)$, weighted by $P_{A|X}(a,x)$ and averaged over $a$. 

The entropy \eqref{eq:towards_rel_ent_assemb}, which depends on $x$, does not measure the distinguishability between the assemblages $\rho_{A|X}$ and $\rho'_{A|X}$. Since the latter are conditional objects, i.e., with inputs, a general strategy to distinguish them must allow for Alice choosing the input for which the assemblages' outputs are optimally distinguishable. Furthermore, Bob can first apply a generalised measurement on his subsystem and communicate the outcome $\gamma$ to her, which she can then use for her input choice. This is the most general procedure within the allowed 1W-LOCCs.
Hence, a generic distinguishing strategy under 1W-LOCCs involves probabilistically chosen inputs that depend on $\gamma$. 
Note, in addition, that the statistics of $\gamma$ generated, described by distributions $P_{\Gamma}$ or $P'_{\Gamma}$, encode differences between $\rho_{A|X}$ and $\rho'_{A|X}$ too and must therefore also be accounted for by a distinguishability measure. The following definition incorporates all these considerations. 
%According to the allowed operations in a general 1W-LOCC protocol, such input distribution can be chosen in a general strategy depending on Bob's output $\gamma$ obtained after performing a generalised quantum measurement on his quantum state. Note also that such quantum measurement generates an output distribution $P_{\Gamma}(\gamma)$ --$P'_{\Gamma}(\gamma)$ if obtained by measuring on $\hat{\rho}'_{A|X}$-- that is itself useful to distinguish both assemblages. 
%\ro{These ideas are precisely encoded in the following proposed definition of relative entropy, so that it fulfills the key properties of a measure of distinguishability under general 1W-LOCCs \cite{footnote}.} 
%
\begin{dfn}[Relative entropy between assemblages]
\label{def:rel_ent_assemb}
Given any two assemblages $\rho_{A|X}$ and $\rho'_{A|X}$, we define the \emph{assemblage relative entropy} of $\rho_{A|X}$ with respect to $\rho'_{A|X}$ as
\begin{widetext}
\begin{align}
%\nonumber S_\text{A}(\rho_{A|X}\|\rho'_{A|X}):=\max_{P_{X|\Gamma},E_{\gamma}} \bigg[S_\text{C}(P_{\Gamma}\|P'_{\Gamma})+\sum_{\gamma,x}P_{X|\Gamma}(x,\gamma)P_{\Gamma}(\gamma)&\\
\label{eqmat:relativeentropy}
%S_\text{Q}\bigg( \frac{\id \otimes E_{\gamma}\hat{\rho}_{A|X}(x)\id \otimes E_{\gamma}^{\dagger}}{P_{\Gamma}(\gamma)}\: \bigg\| \:\frac{\id \otimes E_{\gamma}\hat{\rho}'_{A|X}(x)\id \otimes E_{\gamma}^{\dagger}}{P'_{\Gamma}(\gamma)} \bigg)\bigg]&
%\label{eqmat:relativeentropy}
S_\text{A}(\rho_{A|X}\|\rho'_{A|X}):=\max_{P_{X|\Gamma},\{E_{\gamma}\}} \bigg[S_\text{C}(P_{\Gamma}\|P'_{\Gamma})+\sum_{\gamma,x}P(x|\gamma) P_{\Gamma}(\gamma)
S_\text{Q}\bigg( \frac{\id \otimes E_{\gamma}\hat{\rho}_{A|X}(x)\id \otimes E_{\gamma}^{\dagger}}{P_{\Gamma}(\gamma)}\: \bigg\| \:\frac{\id \otimes E_{\gamma}\hat{\rho}'_{A|X}(x)\id \otimes E_{\gamma}^{\dagger}}{P'_{\Gamma}(\gamma)} \bigg)\bigg],
\end{align}
\end{widetext}
where $E_{\gamma}:\mathcal{H}_B\to\mathcal{H}_B$ are generalised-measurement operators such that $\sum_{\gamma}E^{\dagger}_{\gamma}E_{\gamma}=\id$, $P_{X|\Gamma}$ is a conditional probability distribution of $x$ given $\gamma$, the short-hand notation $P(x|\gamma)\coloneqq P_{X|\Gamma}(x,\gamma)$ has been used, %. In turn, the joint distribution $P_{X,\Gamma}$ is defined by $P_{X,\Gamma}\lo{\coloneqq}P_{\Gamma}P_{X|\Gamma}$, with 
and
\begin{subequations}
\label{eq:Prob_gamma}
\begin{align}
P_{\Gamma}(\gamma)&\coloneqq\Tr [\id \otimes E_{\gamma}\hat{\rho}_{A|X}(x) \id \otimes E^{\dagger}_{\gamma}]=\Tr_B [E_{\gamma} \varrho_{B} E_{\gamma}^{\dagger}],\\
P'_{\Gamma}(\gamma)&\coloneqq\Tr [\id \otimes E_{\gamma} \hat{\rho}'_{A|X}(x) \id \otimes E^{\dagger}_{\gamma}]=\Tr_B [E_{\gamma} \varrho'_{B} E_{\gamma}^{\dagger}],
\end{align}
\end{subequations}
where $\varrho'_{B}$ is Bob's reduced state for the assemblage $\rho'_{A|X}$. 
\end{dfn}
In App. \ref{sec:apprelativeentropy}, we show that $S_\text{A}$ does not increase --on average-- under deterministic 1W-LOCCs and, as its quantum counterpart $S_\text{Q}$, is jointly convex. Hence, $S_\text{A}$ is a proper measure of  distinguishability between assemblages under 1W-LOCCs \cite{footnote0}. The first term inside the maximisation in Eq. \eqref{eqmat:relativeentropy} accounts for the distinguishability between the distributions of measurement outcomes $\gamma$ and the second one for that between the distributions of Alice's outputs and Bob's states resulting from each $\gamma$, averaged over all inputs and measurement outcomes. In turn, the maximisation over $\{E_{\gamma}\}$ and $P_{X|\Gamma}$ ensures that these output distributions and states are distinguished using the optimal 1W-LOCC-compatible strategy.

We are now in a good position to introduce a convex steering monotone. We do it with a theorem.

\begin{thm}[1W-LOCC-monotonicity and convexity of $\mathscr{S}_{\text{R}}$] 
\label{theo: rel_ent}
The \emph{relative entropy of steering} $\mathscr{S}_{\text{R}}$, defined for an assemblage $\rho_{A|X}$ as 
\begin{align}
\label{eq:def_rel_ent_steering}
\mathscr{S}_{\text{R}}(\rho_{A|X}):=\min_{\sigma_{A|X} \in \mathsf{LHS}}S_\text{A}(\rho_{A|X}\parallel\sigma_{A|X}),
\end{align}
is a convex steering monotone.
\end{thm}
\noindent The theorem is proven in App. \ref{sec:relativeentropy}.

%%%%%%%%%%%%%%%%%%%%%%%%%%%%%%%%%%%%%%%%%%%%%%%%%%%%%%%%%%%%%
%%%%%%%%%%%%%%%%%%%%%%%%%%%%%%%%%%%%%%%%%%%%%%%%%%%%%%%%%%%%%
\section{Other convex steering monotones}
\label{sec:other_measures}
Apart from $\mathscr{S}_{\text{R}}$ two other quantifiers of steering have been recently proposed: the steerable weight \cite{Skrzypczyk13} and the robustness of steering \cite{Piani15}. In this section, we show that these are also convex steering monotones.
\begin{dfn}[Steerable weight \cite{Skrzypczyk13}] The steerable weight $\mathscr{S}_{\rm{W}}(\rho_{A|X})$ of a normalised assemblage $\rho_{A|X}$ is the minimum $\nu\in\mathbb{R}_{\geq0}$ such that
\begin{equation}
\label{eq:decompsw}
\rho_{A|X}= \nu\, \tilde{\rho}_{A|X} + (1-\nu) \sigma_{A|X},
\end{equation}
with $\tilde{\rho}_{A|X}$ an arbitrary normalised assemblage and normalised $\sigma_{A|X}\in \mathsf{LHS}$.
\end{dfn}

\begin{dfn}[Robustness of steering \cite{Piani15}] The robustness of steering $\mathscr{S}_{\rm{Rob}}(\rho_{A|X})$ of a normalised assemblage $\rho_{A|X}$ is the minimum $\nu\in\mathbb{R}_{\geq0}$ such that the normalised assemblage
\begin{equation}
\label{eq:defrobust}
\sigma_{A|X}\coloneqq \frac{1}{1+\nu} \rho_{A|X}+\frac{\nu}{1+\nu}\, \tilde{\rho}_{A|X}
\end{equation}
belongs to $\mathsf{LHS}$, with $\tilde{\rho}_{A|X}$ an arbitrary normalised assemblage.
\label{def:Rob}
\end{dfn}

In App. \ref{app:other_measures}, we prove the following theorem.
\begin{thm}[1W-LOCC-monotonicity and convexity of $\mathscr{S}_{\rm{W}}$ and  $\mathscr{S}_{\rm{Rob}}$]  Both $\mathscr{S}_{\rm{W}}$ and  $\mathscr{S}_{\rm{Rob}}$ are convex steering monotones. 
\label{theo:W_and_Rob}
\end{thm}

To end up with, we note that a steering measure for assemblages containing continuous-variable (CV)  bosonic systems in Gaussian states has very recently appeared \cite{Kogias14}. Even though our formalism can be straightforwardly extended to CV systems, such extension is outside the scope of the present paper.

\section{Assemblage conversions and no steering bits}\label{sec:assemblageconversions}
We say that $\Psi_{A|X}$ and $\Psi'_{A|X}$ are \emph{pure assemblages} if they are of the form 
\begin{subequations}
\begin{align}
\label{eq:pure_ensembledef}
\Psi_{A|X}&\coloneqq\{P_{A|X}(a,x),\ketbra{\psi(a,x)}{\psi(a,x)}\}_{a,x},\\
\Psi'_{A|X}&\coloneqq\{P'_{A|X}(a,x),\ketbra{\psi'(a,x)}{\psi'(a,x)}\}_{a,x},
\end{align}
\end{subequations} 
where $\ket{\psi(a,x)}$ and $\ket{\psi'(a,x)}\in\mathcal{H}_B$, and \emph{pure orthogonal assemblages} if, in addition, $\braket{\psi(a,x)}{\psi(\tilde{a},x)}=\delta_{a\, \tilde{a}}=\braket{\psi'(a,x)}{\psi'(\tilde{a},x)}$ for all $x$%, being  $\delta_{a\tilde{a}}$ the Kronecker delta
. Note that pure orthogonal assemblages are the ones obtained when Alice and Bob share a pure maximally entangled state and Alice performs a von-Neumann measurement on her share. We present two theorems about assemblage conversions under 1W-LOCCs. 

The first one, proven in App.~\ref{sec:assemb_transformations}, establishes necessary and sufficient conditions for stochastic-1W-LOCC conversions between pure orthogonal assemblages, therefore playing a similar role here to the one played in entanglement theory by Vidal's theorem  \cite{Vidal99} for stochastic-LOCC pure-state conversions.
\begin{thm}[Criterion for stochastic-1W-LOCC conversion]
\label{thm:aseemb_tranform}
Let $\Psi_{A|X}$ and $\Psi'_{A|X}$ be any two pure orthogonal assemblages with $d=s=r=2$. Then, $\Psi_{A|X}$ can be transformed into $\Psi'_{A|X}$ by a stochastic 1W-LOCC iff: either $\Psi'_{A|X}\in\mathsf{LHS}$ or $P'_{A|X}=P_{A|X}$ and 
\begin{equation}
\label{eq:overlap}
 |\braket{\psi'(a,0)}{\psi'(a,1)}|=|\braket{\psi(a\oplus\alpha,0)}{\psi(a\oplus\alpha,1)}|\ \forall\ a,
\end{equation}
for some $\alpha\in\{0,1\}$.
\end{thm}
\noindent In other words, no pure orthogonal assemblage of minimal dimension can be obtained via a 1W-LOCC, not even probabilistically, from a pure orthogonal assemblage of minimal dimension with a different state-basis overlap (except for trivial relabellings of $a$, given by $\alpha$) unless the former is unsteerable.
Hence, each state-basis overlap defines an inequivalent class of steering, there being infinitely many of them. This is in a way reminiscent to the inequivalent classes of entanglement in multipartite \cite{Duer99} or infinite-dimensional bipartite \cite{Owari04} systems, but here the phenomenon is found already for bipartite systems of minimal dimension.

The second theorem, proven in App. \ref{sec:no_steering_bits}, rules out the possibility of there being a (non-orthogonal) minimal-dimension pure assemblage from which all assemblages can be obtained. 
\begin{thm}[Non-existence of steering bits]
\label{thm:nosteeringbig}
There exists no pure assemblage with $d=s=r=2$ that can be transformed into any assemblage by stochastic 1W-LOCCs. 
\end{thm}
\noindent Hence, among the minimal-dimension assemblages there is no operationally well defined \emph{unit of steering}, or \emph{steering bit}, i.e., an assemblage from which all assemblages can be obtained for free and can therefore be taken as a measure-independent maximally steerable assemblage. This is again in striking contrast to entanglement theory, where pure maximally entangled states can be defined without the need of entanglement quantifiers and each one can be transformed into any state by deterministic LOCCs \cite{Vidal99,Nielsen99}.

%%%%%%%%%%%%%%%%%%%%%%%%%%%%%%%%%%%%%%%%%%%%%%%%%%%%%%%%%%%%%%%%%%%%%%%%%%%%%%%%%
%%%%%%%%%%%%%%%%%%%%%%%%%%%%%%%%%%%%%%%%%%%%%%%%%%%%%%%%%%%%%%%%%%%%%%%%%%%%%%%%%
%%%%%%%%%%%%%%%%%%%%%%%%%%%%%%%%%%%%%%%%%%%%%%%%%%%%%%%%%%%%%%%%%%%%%%%%%%%%%%%%%

\section{Discussion and outlook}
\label{sec:discussionandoutlook}
We have introduced the resource theory of Einstein-Podolsky-Rosen steering. The free operations of the theory are the 1W-LOCCs from the quantum part to the black box, i.e.,  
all the assemblage transformations involving deterministic bit wirings on Alice's side and stochastic quantum operations on Bob's assisted by 1-way classical communication from Bob to Alice.
These operations satisfy the basic requirement of mapping all unsteerable assemblages into unsteerable assemblages and are, besides, also the allowed operations that naturally arise from the basic security constraints of one-sided device-independent QKD, where steering is a physical resource. With these operations, we introduced the notion of convex steering monotones, presented the relative entropy of steering as a convenient example thereof, and proved monotonicity and convexity of two other previously proposed steering measures. In addition, for minimal-dimensional systems, we established necessary and sufficient conditions for stochastic-1W-LOCC conversions between pure-state assemblages and proved the non-existence of steering bits.

It is instructive to emphasise that the derived 1W-LOCCs are hybrids between the operations that map separable states into separable states, stochastic LOCCs, and those that map Bell local correlations into Bell local correlations, local wirings assisted by prior-to-input classical communication \cite{Gallego12,Navascues14,deVicente14}. In fact,  %Regarding the latter, a resource-theory approach to Bell non-locality is only partially developed \cite{Gallego12,Navascues14,deVicente14}. Hence, 
our findings are also potentially useful for the quantification of Bell non-locality. 
In addition, our work offers a number of challenges for future research. Namely, for example, the non-existence of steering bits of minimal dimension can be seen as an impossibility of steering dilution of minimal-dimension assemblages in the single-copy regime. We leave as open questions what the rules for steering dilution and distillation are for higher-dimensional systems, mixed-state assemblages, or in asymptotic multi-copy regimes, and what the steering classes are for mixed-state assemblages. Moreover, other fascinating questions are whether one can formulate a notion of bound steering or an analogue to the positive-partial-transpose criterion for assemblages.

%%%%%%%%%%%%%%%%%%%%%%%%%%%%%%%%%%%%%%%%%%%%%%%%%%%%%%%%%%%%%%%%%%%%%%%%%%%%%%%%%
%\subsection{The quantum-state representation of the ensemble}

\section*{Acknowledgements}

We would like to thank Antonio Ac\'in, Daniel Cavalcanti, Paul Skrzypczyk and Marco T\'ulio Quintino for discussions and the EU (RAQUEL, SIQS) for support. RG acknowledges support from the Alexander von Humboldt Foundation and LA from the EU (REQS - Marie Curie IEF No 299141).

\appendix

\section{Parametrisation of the class $\mathsf{1WLOCC}$}
\label{sec:genericoperations}
In this appendix, we show that any generic assemblage map $\mathcal{M}$ involving stochastic local quantum operations on Bob's side, one-way classical communication from Bob to Alice, and deterministic (probability-preserving) local wirings on Alice's side is of the form given by Eqs. \eqref{eq:final_assemb} and \eqref{eq:def_W} and, therefore, belongs to the class the class $\mathsf{1WLOCC}$ of Definition \ref{Def1WLOCCs}. 

Without loss of generality, such $\mathcal{M}$ can be decomposed into the following sequence (see Fig. \ref{fig:1}) of operations:

\begin{enumerate}
\item Bob applies an arbitrary stochastic generalised measurement, described by a  completely-positive non trace-preserving  map $\mathcal{E}%:\mathcal{L}(\mathcal{H}_B)\to\mathcal{L}({\mathcal{H}_B}_f)
$, defined by Eqs. \eqref{eq:Kraus_def}, to his quantum subsystem before Alice introduces an input to her device. 
Note that, since the non-signalling condition \eqref{eq:nosignaling} is fulfilled,  Bob has a well defined reduced quantum state $\rho_{B}$, given by Eq. \eqref{eq:nosignaling}, independently of Alice still not having  chosen her measurement input $x$.
Therefore, Bob's measurement gives the outcome $\omega$ with the $x$-independent probability $P_\Omega(\omega)$ given by Eq. \eqref{eq:P_Omega}.

\item Bob sends the outcome $\omega$ to Alice. Alice applies a local wiring, described by the normalised conditional probability distribution $P_{X|X_f,\Omega}$, to the input $x_f$ of the final device and to $\omega$, and uses the output of this wiring as the input $x$ of her initial device. For a given $x$, her initial device outputs $a$ with a probability determined by the conditional distribution $P_{A|X}$ of the initial assemblage. In that case, Bob's normalized state is given by the $x_f$-independent density operator
\begin{equation}\label{eq:condstate}
\varrho(a,x,\omega,x_f)\coloneqq\frac{K_{\omega} \varrho(a,x) K^{\dagger}_{\omega}}{\Tr[K_{\omega} \varrho(a,x) K^{\dagger}_{\omega}]}.
\end{equation}

\item  Alice applies a local wiring, described by the normalised conditional probability distribution $P_{A_f|A,X, \Omega, X_f}$, to all the previously generated classical bits, $a$, $x$, $\omega$, and $x_f$, and uses the output of this wiring as the output $a_f$ of her final device.
This final processing of the bit $a_f$ does not affect Bob's state. Thus, Bob's system ends up in the state $\varrho(a,x,\omega,x_f,a_f)\coloneqq\varrho(a,x,\omega,x_f)$.
\end{enumerate}

We denote by $P^i_{\Omega|A,X}$ the conditional distribution of $\Omega$ given $A$ and $X$, for $X$ chosen independently of $\Omega$ (in contrast to Step 2 above), with elements $P^i_{\Omega|A,X}(\omega,a,x)\coloneqq\Tr[K_{\omega} \rho(a,x) K^{\dagger}_{\omega}]$. With this, the components $\varrho_{A_f|X_f}(a_f,x_f)$ of the final assemblage $\rho_{A_f|X_f}$ are explicitly given by:
\begin{widetext}
\begin{align}
\nonumber
\varrho_{A_f|X_f}(a_f,x_f)&\coloneqq P_{A_f|X_f}(a_f,x_f)\times\varrho_{f}(a_f,x_f)\\
\label{eq:prot1}
&=\sum_{a,x,\omega}P_{A_f,A,X,\Omega|X_f}(a_f,a,x,\omega,x_f)\times \varrho(a_f,a,x,\omega,x_f)\\
\label{eq:prot2}
&=\sum_{a,x,\omega}P_{A_f,A,X,\Omega|X_f}(a_f,a,x,\omega,x_f)\times \frac{K_{\omega} \varrho(a,x) K^{\dagger}_{\omega}}{\Tr[K_{\omega} \varrho(a,x) K^{\dagger}_{\omega}]}\\
\label{eq:prot3}
&=\sum_{a,x,\omega}P_{\Omega}(\omega)P_{X|X_f,\Omega}(x,x_f,\omega)P_{A|X,\Omega}(a,x,\omega)P_{A_f|A,X,\Omega,X_f}(a_f,a,x,\omega,x_f)\times \frac{K_{\omega}\varrho(a,x) K^{\dagger}_{\omega}}{P^i_{\Omega|A,X}(\omega,a,x)}.
\end{align}
 \end{widetext}
Eq. \eqref{eq:prot1} follows from basic properties of probability distributions and ensembles of states. Eq. \eqref{eq:prot2} follows from the definition of $\rho(a,x,\omega,x_f,a_f)$. Eq. \eqref{eq:prot3} follows from Bayes' theorem together with the facts that $P_{A|X,\Omega,X_f}=P_{A|X,\Omega}$ (the output of Alice's initial device only depends on the input $x$ and the measurement outcome $\omega$) and $P_{\Omega|X_f}=P_{\Omega}$ (the measurement outcome $\omega$ is independent of the input of Alice's final device), and from the definition of $P^i_{\Omega|A,X}$. Next, note that, since the statistics of $A$ is fully determined by $X$ and $\Omega$ regardless of whether $X$ and $\Omega$ are independent or not, it holds that 
\begin{align}
\nonumber
P_{A|X,\Omega}&=P^i_{A|X,\Omega}\\
\nonumber
&=\frac{P^i_{A, X,\Omega}}{P^i_{X|\Omega}P^i_{\Omega}}\\
\nonumber
&=\frac{P^i_{A, X, \Omega}}{P^i_{X}P_{\Omega}}\\
\nonumber
&=\frac{P^i_{\Omega|A, X }P^i_{A, X }}{P^i_{X}P_{\Omega}}\\
&=\frac{P^i_{\Omega|A, X }P^i_{A|X }}{P_{\Omega}}
\label{eq:Pi_Bayes},
\end{align}
where we have used Bayes' theorem, that $P^i_{\Omega}=P_{\Omega}$, and that, by definition, $P^i_{X|\Omega}=P^i_{X}$.
Inserting Eq.  \eqref{eq:Pi_Bayes} into Eq. \eqref{eq:prot3}, we obtain
 \begin{widetext}
\begin{align}
\nonumber
\varrho_{A_f|X_f}(a_f,x_f)&=\sum_{a,x,\omega}
P_{X|X_f,\Omega}(x,x_f,\omega)P^{i}_{A|X}(a,x)P_{A_f|A,X,\Omega, X_f}(a_f,a,x,\omega,x_f)\times K_{\omega} \varrho(a,x) K^{\dagger}_{\omega}\\
\label{eq:prot6}
&=\sum_{a,x,\omega}P_{X|X_f,\Omega}(x,x_f,\omega)P_{A_f|A,X,\Omega, X_f}(a_f,a,x,\omega,x_f)\times K_{\omega} \varrho_{A|X}(a,x) K^{\dagger}_{\omega},\ \forall\  (a_f, x_f),
\end{align}
 \end{widetext}
 where \eqref{eq:prot6} follows from the fact that $P^{i}_{A|X}=P_{A|X}$ and the definition of $\varrho_{A|X}$. The right-hand side of Eq. \eqref{eq:prot6} gives the most general expression of the components $\varrho_{A_f|X_f}(a_f,x_f)$ of $\mathcal{M}(\rho_{A|X})$ explicitly as a function of the components $\varrho_{A|X}$ of$\rho_{A|X}$. The reader can straightforwardly verify that the quantum representation $\mathcal{M}(\hat{\rho}_{A|X})$ of the obtained final assemblage $\mathcal{M}(\rho_{A|X})$ is given by the right-hand side of Eq. \eqref{eq:final_assemb}.

\section{Invariance of $\mathsf{LHS}$ under $1W-LOCC$ maps}
\label{sec:proof_Theorem1WLOCCs}
We now show that if $\rho_{A|X}\in\mathsf{LHS}$ then, for all $\mathcal{M}\in\mathsf{1WLOCC}$,  $\mathcal{M}(\rho_{A|X})\in\mathsf{LHS}$. 
\begin{proof}[Proof of Theorem \ref{Theorem1WLOCCs}]
Replacing $\rho_{A|X}$ in Eq. \eqref{eq:prot6} by the right-hand side of Eq. \eqref{eq:lhs}, we write
\begin{widetext}
\begin{align}
\nonumber
\varrho_{A_f|X_f}&=\sum_{a,x,\omega,\lambda}P_{\Lambda}(\lambda)P_{A|X,\Lambda}(a,x,\lambda)P_{X|X_f,\Omega}(x,x_f,\omega)P_{A_f|A,X,\Omega, X_f}(a_f,a,x,\omega,x_f)\times K_{\omega} \xi(\lambda) K^{\dagger}_{\omega}\\
\label{eq:nostee1}&=\sum_{a,x,\omega,\lambda}P_{\Lambda}(\lambda)P_{\Omega|\Lambda}(\omega,\lambda)P_{A|X,\Lambda}(a,x,\lambda)P_{X|X_f,\Omega}(x,x_f,\omega)P_{A_f|A,X,\Omega, X_f}(a_f,a,x,\omega,x_f)\times \xi(\lambda,\omega),
\end{align}
\end{widetext}
where  the conditional probability $P_{\Omega|\Lambda}(\omega,\lambda)\coloneqq\Tr[K_{\omega} \xi(\lambda) K^{\dagger}_{\omega}]$ and the normalized state $\xi(\lambda,\omega)\coloneqq\frac{K_{\omega} \xi(\lambda) K^{\dagger}_{\omega}}{P_{\Omega|\Lambda}(\omega,\lambda)}$ have been introduced. Using that $a_f$ does not explicitly depend on $\lambda$, we see that 
\begin{equation}
\label{eq:neweq}
P_{A_f|A,X,\Omega,X_f}=P_{A_f|A,X,\Omega, \Lambda, X_f}.
\end{equation}
In turn, using the facts that $x$ is independent of $\lambda$ and $a$ depends only on $x$ and $\lambda$, and Bayes' theorem, we see that 
\begin{equation}
\label{eq:nostee6}
P_{A|X,\Lambda}P_{X|X_f,\Omega}=P_{A|X,\Omega,\Lambda ,X_f}P_{X|\Omega,\Lambda,X_f}=P_{A,X|\Omega,\Lambda,X_f}.
\end{equation}
Substituting into Eq. \eqref{eq:nostee1} yields
\begin{widetext}
\begin{align}
\nonumber
\rho_{A_f|X_f}
%&=\sum_{a,x,\omega,\lambda}P_{\Lambda}(\lambda)P_{\Omega|\Lambda}(\omega,\lambda)P_{A|X,\Lambda}(a,x,\lambda)P_{X|X_f,\Omega}(x,x_f,\omega)P_{A_f|A,X,\Omega,\Lambda,X_f}(a_f,a,x,\omega,\lambda,x_f))\times \xi(\lambda,\omega) \\
%\label{eq:nostee3}
&=\sum_{a,x,\omega,\lambda}P_{\Lambda}(\lambda)P_{\Omega|\Lambda}(\omega,\lambda)P_{A,X|\Omega,\Lambda,X_f}(a,x,\omega,\lambda,x_f)P_{A_f|A,X,\Omega,\Lambda,X_f}(a_f,a,x,\omega,\lambda,x_f)\times \xi(\lambda,\omega)\\
\label{eq:nostee4}
&=\sum_{\omega,\lambda}P_{\Omega,\Lambda}(\omega,\lambda)P_{A_f|\Omega,\Lambda,X_f}(a_f,\omega,\lambda,x_f)\times \xi(\lambda,\omega)\\
\label{eq:nostee5}
&=\sum_{\tilde{\lambda}}P_{\tilde{\Lambda}}\left(\tilde{\lambda}\right)P_{A_f|X_f,\tilde{\Lambda}}\left(a_f,x_f,\tilde{\lambda}\right)\:\sigma\left(\tilde{\lambda}\right),
\end{align}
\end{widetext}
where Eq. \eqref{eq:nostee4} follows from Bayes' theorem and summing over $x$ and $a$, and  Eq. \eqref{eq:nostee5} follows from defining the hidden variable $\tilde{\lambda}\coloneqq(\omega,\lambda)$ governed by the normalized probability distribution $P_{\tilde{\Lambda}}\coloneqq P_{\Omega,\Lambda}$. Eq. \eqref{eq:nostee5} manifestly shows that $\rho_{A_f|X_f}\in\mathsf{LHS}$. 
\end{proof}

\section{The relative entropy of steering} \label{sec:apprelativeentropy}
\label{sec:relativeentropy}
In this appendix we prove Theorem \ref{theo: rel_ent}. The proof strategy is similar to that of the proof that the relative entropy of entanglement for quantum states is a convex entanglement monotone \cite{Vedral98}. It relies on two Lemmas, which we  state next but whose proofs we leave for App. \ref{Sec:Proofs_of_Lemmas}. 

\begin{lem}
\label{lem:strongmonotonicitydistance}
The assemblage relative entropy $S_\text{A}$, defined by Eq. \eqref{eqmat:relativeentropy}, does not increase, on average, under deterministic 1W-LOCCs. That is, for any map $\mathcal{M}$ of the form given by Eqs. \eqref{eq:final_assemb} and \eqref{eq:def_W} but with $\sum_{\omega}K^{\dagger}_{\omega}K_{\omega}= \id$ and any two assemblages $\rho_{A|X}$ and $\rho'_{A|X}$, $S_\text{A}$ satisfies the inequality
\begin{eqnarray}
\label{eq:strongmonotonicitydist}
&&\sum_{\omega} P_{\Omega}(\omega) S_\text{A}\left( \frac{\mathcal{M}_{\omega}\left(\rho_{A|X}\right)}{\Tr\left[\mathcal{M}_{\omega}(\rho_{A|X})\right]}\Bigg\|\frac{\mathcal{M}_{\omega}\left(\rho'_{A|X}\right)}{\Tr\left[\mathcal{M}_{\omega}(\rho'_{A|X})\right]}\right)\nonumber\\
&& \hspace{2cm}\leq S_\text{A}\left(\rho_{A|X}|\rho'_{A|X}\right),
\end{eqnarray}
where $\mathcal{M}_{\omega}$ is the stochastic map defined in Eq. \eqref{eq:def_M_mu}, $P_{\Omega}=\Tr\left[\mathcal{M}_{\omega}(\rho_{A|X})\right]$, and $P'_{\Omega}=\Tr\left[\mathcal{M}_{\omega}(\rho'_{A|X})\right]$, with $\sum_{\omega} P_{\omega}=1=\sum_{\omega} P'_{\omega}$.
\end{lem}

\begin{lem}\label{lem:convex}
The assemblage relative entropy $S_\text{A}$, defined by Eq. \eqref{eqmat:relativeentropy}, is jointly convex. That is, given two sets $\{\rho^{(j)}_{A|X}\}_{j=1,\ldots,n}$ and $\{\rho'^{(j)}_{A|X}\}_{j=1,\ldots,n}$ of $n$ arbitrary assemblages each and $n$ positive real numbers $\{\mu^{(j)}\}_{j=1,\ldots,n}$ such that $\sum_{j}\mu^{(j)}=1$, with $n\in\mathbb{N}$, $S_\text{A}$ satisfies the inequality
\begin{eqnarray}
\label{eq:convexity}
&&S_\text{A}\left(\sum_{i}\mu^{(j)} \rho^{(j)}_{A|X} \Big\|\sum_{j}\mu^{(j)} \rho'^{(j)}_{A|X}\right)\nonumber\\
&&\hspace{2cm} \leq \sum_{j} \mu^{(j)} S_\text{A}\left(\rho^{(j)}_{A|X}\big\|\rho'^{(j)}_{A|X}\right).
\end{eqnarray}
\end{lem}

We are now in a good position to prove the theorem.

\begin{proof}[Proof of Theorem \ref{theo: rel_ent}]
That the relative entropy of steering $\mathscr{S}_{\text{R}}$, defined in Eq. \eqref{eq:def_rel_ent_steering}, satisfies condition $i$) follows immediately from its definition  and the positivity of the von-Neumann relative entropy for quantum states. Conditions $ii$) and $iii$), 1W-LOCC monotonicity and convexity of $\mathscr{S}_{\text{R}}$, can be proven in analogous fashion to LOCC monotonicity and convexity of the relative entropy of entanglement, respectively. We include their proofs for completeness.

To prove condition $ii$), we denote  by $\sigma^*$ an unsteerable assemblage for which the minimisation in Eq. \eqref{eq:def_rel_ent_steering} is attained, i.e., such that 
\begin{equation}
S_\text{A}(\rho_{A|X}\|\sigma^*)\coloneqq\mathscr{S}_{\text{R}}(\rho_{A|X})
\end{equation}
 and by $\hat{\sigma}^*_{\mu}$ an unsteerable assemblage such that 
 \begin{equation}
S_\text{A}\left(\frac{\mathcal{M}_{\omega}(\rho_{A|X})}{\Tr\left[\mathcal{M}_{\omega}(\rho_{A|X})\right]}\bigg\|\hat{\sigma}^*_{\mu}\right)\coloneqq\mathscr{S}_{\text{R}}\left( \frac{\mathcal{M}_{\omega}(\rho_{A|X})}{\Tr\left[\mathcal{M}_{\omega}(\rho_{A|X})\right]}\right).
\end{equation}
Then, we write 
\begin{eqnarray}
\nonumber 
&&\sum_{\omega} P_{\Omega}(\omega) \mathscr{S}_{\text{R}}\left( \frac{\mathcal{M}_{\omega}(\rho_{A|X})}{\Tr\left[\mathcal{M}_{\omega}(\rho_{A|X})\right]}\right)\\
&&\hspace{0.25cm}=\sum_{\omega} P_{\Omega}(\omega) S_\text{A}\left(\frac{\mathcal{M}_{\omega}(\rho_{A|X})}{\Tr\left[\mathcal{M}_{\omega}(\rho_{A|X})\right]}\bigg\|\sigma^*_{\mu}\right)\nonumber\\
\label{eq:strongmon1}
&&\hspace{0.25cm}\leq \sum_{\omega}P_{\Omega}(\omega)S_\text{A}\left(\frac{\mathcal{M}_{\omega}(\rho_{A|X})}{\Tr\left[\mathcal{M}_{\omega}(\rho_{A|X})\right]}\bigg\| \frac{\mathcal{M}_{\omega}(\sigma^*)}{\Tr[\mathcal{M}_{\omega}(\sigma^*)]} \right)\\
\label{eq:strongmon2}
&&\hspace{0.25cm}\leq S_\text{A}(\rho_{A|X}\|\sigma^*)\\
\label{eq:last}
&&\hspace{0.3cm}=\mathscr{S}_{\text{R}}(\rho_{A|X}),
\end{eqnarray}
where Eq. \eqref{eq:strongmon1} follows because $\sigma^*_{\mu}$ minimises the assemblage relative entropy in each $\omega$-th term in the sum and $\frac{\mathcal{M}_{\omega}(\sigma^*)}{\Tr[\mathcal{M}_{\omega}(\sigma^*)]} \in \mathsf{LHS}$, Eq. \eqref{eq:strongmon2} due to Lemma \ref{lem:strongmonotonicitydistance}, and Eq. \eqref{eq:last} due to the definition of $\sigma^*$. 

To prove condition condition $iii$), we further introduce unsteerable assemblages $\sigma'^*$ and $\sigma_{\text{mix}}^*$ such that 
\begin{equation}
S_\text{A}(\rho'_{A|X}\|\sigma'^*)=\mathscr{S}_{\text{R}}(\rho'_{A|X})
\end{equation}
 and 
\begin{eqnarray}
&&S_\text{A}\left(\mu\, \rho_{A|X}+(1-\mu)\rho'_{A|X}\|\sigma_{\text{mix}}^*\right)\nonumber\\
&&\hspace{1.5cm}=\mathscr{S}_{\text{R}}\left(\mu\, \rho_{A|X}+(1-\mu)\rho'_{A|X}\right).
\end{eqnarray}
Then, we write
\begin{eqnarray}
\nonumber 
&&\mu\, \mathscr{S}_{\text{R}}\left(\rho_{A|X}\right)+(1-\mu)\mathscr{S}_{\text{R}}\left(\rho'_{A|X}\right)\\
\nonumber &=&\mu\, S_\text{A}(\rho_{A|X}\|\sigma^*)+(1-\mu)S_\text{A}(\rho'_{A|X}\|\sigma'^*)\\
\label{eq:use_Lemma2}
&\geq& S_\text{A}\left(\mu\,\rho_{A|X}+(1-\mu)\rho'_{A|X}\|\mu\,\sigma^*+(1-\mu)\sigma'^*\right)\\
\label{eq:use_minimization}
&\geq& S_\text{A}\left(\mu\,\rho_{A|X}+(1-\mu)\rho'_{A|X}\|\sigma^*_{\text{mix}}\right)\\
\label{eq:last2}
&\eqqcolon &\mathscr{S}_{\text{R}}\left(\mu\,\rho_{A|X}+(1-\mu)\rho'_{A|X}\right),
\end{eqnarray}
where Eq. \eqref{eq:use_Lemma2} holds due to Lemma \ref{lem:convex}, Eq. \eqref{eq:use_minimization} because $\sigma^*_{\text{min}}$ minimizes the corresponding assemblage relative entropy and $\mu \sigma* + (1-\mu)\sigma'^* \in \mathsf{LHS}$, and Eq. \eqref{eq:last2} by the definition of $\sigma^*_{\text{min}}$. 
\end{proof}

%%%%%%%%%%%%%%%%%%%%%%%%%%%%%%%%%%%%%%%%%%%%%%%%%%%%%%%%%%%%%
%%%%%%%%%%%%%%%%%%%%%%%%%%%%%%%%%%%%%%%%%%%%%%%%%%%%%%%%%%%%%
\section{Proof of Theorem \ref{theo:W_and_Rob}}
\label{app:other_measures}
\begin{proof}[Proof of Theorem \ref{theo:W_and_Rob}]
Let us first prove the theorem's statement concerning the steerable weight. That $\mathscr{S}_{\rm{W}}$ satisfies condition $i$) of Definition \ref{def:Mon_Conv} follows immediately from its definition. 
To prove that it fulfils  condition $ii$), first, we apply the map $\mathcal{M}_{\omega}$ to both sides of Eq. \eqref{eq:decompsw} and renormalize. This yields
\begin{align}
\label{eq:decompsw2}
 \nonumber
 \frac{\mathcal{M}_{\omega}(\rho_{A|X})}{\Tr\left[\mathcal{M}_{\omega}(\rho_{A|X})\right]}&= \nu\, \frac{\mathcal{M}_{\omega}\left(\tilde{\rho}_{A|X}\right)}{\Tr\left[\mathcal{M}_{\omega}(\rho_{A|X})\right]}\\
 \nonumber
 &+ (1-\nu) \frac{\mathcal{M}_{\omega}\left(\sigma_{A|X}\right)}{\Tr\left[\mathcal{M}_{\omega}(\rho_{A|X})\right]}\\
  \nonumber
  &= \nu\, \frac{\mathcal{M}_{\omega}\left(\tilde{\rho}_{A|X}\right)}{\Tr\left[\mathcal{M}_{\omega}(\tilde{\rho}_{A|X})\right]}\frac{\Tr\left[\mathcal{M}_{\omega}(\tilde{\rho}_{A|X})\right]}{\Tr\left[\mathcal{M}_{\omega}(\rho_{A|X})\right]}\\
&+ (1-\nu) \frac{\mathcal{M}_{\omega}\left(\sigma_{A|X}\right)}{\Tr\left[\mathcal{M}_{\omega}(\sigma_{A|X})\right]}\frac{\Tr\left[\mathcal{M}_{\omega}(\sigma_{A|X})\right]}{\Tr\left[\mathcal{M}_{\omega}(\rho_{A|X})\right]}
.
\end{align}
Denoting by $\nu^{*}_\omega$ the minimum $\nu\in\mathbb{R}_{\geq0}$ such that a decomposition of the form of Eq. \eqref{eq:decompsw2} is possible, it is clear that 
\begin{equation}
\label{eq:wth_term_asterix}
\nu^{*}_\omega\leq\mathscr{S}_{\rm{W}}(\rho_{A|X}),
\end{equation}
as any $\nu\in\mathbb{R}_{\geq0}$ that allows for a decomposition as in Eq. \eqref{eq:decompsw} allows also for one as in Eq. \eqref{eq:decompsw2}.
Furthermore, taking into account that, since $\sigma_{A|X}\in \mathsf{LHS}$, it holds that $\frac{\mathcal{M}_{\omega}\left(\sigma_{A|X}\right)}{\Tr\left[\mathcal{M}_{\omega}(\sigma_{A|X})\right]}\in \mathsf{LHS}$, it is also clear that 
\begin{equation}
\label{eq:wth_term}
\mathscr{S}_{\rm{W}}\left(\frac{\mathcal{M}_{\omega}(\rho_{A|X})}{\Tr\left[\mathcal{M}_{\omega}(\rho_{A|X})\right]}\right)\leq\nu^{*}_\omega \frac{\Tr\left[\mathcal{M}_{\omega}(\tilde{\rho}_{A|X})\right]}{\Tr\left[\mathcal{M}_{\omega}(\rho_{A|X})\right]}.
\end{equation}
Hence, we obtain
\begin{eqnarray}
\nonumber
&&\sum_{\omega}P_{\Omega}(\omega)\mathscr{S}_{\rm{W}}\left(\frac{\mathcal{M}_{\omega}(\rho_{A|X})}{\Tr\left[\mathcal{M}_{\omega}(\rho_{A|X})\right]}\right)\\
\nonumber 
 &\leq&\sum_{\omega}P_{\Omega}(\omega)\nu^{*}_\omega \frac{\Tr\left[\mathcal{M}_{\omega}(\tilde{\rho}_{A|X})\right]}{\Tr\left[\mathcal{M}_{\omega}(\rho_{A|X})\right]}\\
 \nonumber 
  &\leq&\sum_{\omega}  \Tr\left[\mathcal{M}_{\omega}(\tilde{\rho}_{A|X})\right]\mathscr{S}_{\rm{W}}(\rho_{A|X})\\
 &=&\mathscr{S}_{\rm{W}}(\rho_{A|X}),
\end{eqnarray}
where the second inequality is due to the facts that $P_{\omega}=\Tr\left[\mathcal{M}_{\omega}(\rho_{A|X})\right]$ and \eqref{eq:wth_term_asterix} and the last equality to the fact that, since $\tilde{\rho}_{A|X}$ is a normalised assemblage and $\mathcal{M}$ is deterministic map, $\sum_{\omega} \Tr\left[\mathcal{M}_{\omega}(\tilde{\rho}_{A|X})\right]=1$.

To prove the validity of condition $iii$) for $\mathscr{S}_{\rm{W}}$, we first write 
\begin{align}
\nonumber
&\mu\, \rho_{A|X}+(1-\mu)\rho'_{A|X}\\&=\nonumber \mu\, \left[\mathscr{S}_{\rm{W}}\left(\rho_{A|X}\right)\tilde{\rho}_{A|X}+\left(1-\mathscr{S}_{\rm{W}}\left(\rho_{A|X}\right)\right)\sigma_{A|X}\right]\\
\label{eq:first_eq}
&+(1-\mu)\left[\mathscr{S}_{\rm{W}}\left(\rho'_{A|X}\right)\tilde{\rho}'_{A|X}+\left(1-\mathscr{S}_{\rm{W}}\left(\rho_{A|X}\right)\right)\sigma'_{A|X}\right]\\
\label{eq:second_eq}
&=\nu^{(\mu)}\, \tilde{\rho}^{(\mu)}_{A|X} + \left(1-\nu^{(\mu)}\right) \sigma^{(\mu)}_{A|X},
\end{align}
where Eq. \eqref{eq:first_eq} holds due to the definition of $\mathscr{S}_{\rm{W}}$ and, in Eq. \eqref{eq:second_eq}, we have introduced the positive real 
\begin{equation}
\label{eq:def_nu_mu}
\nu^{(\mu)}\coloneqq\mu\, \mathscr{S}_{\rm{W}}\left(\rho_{A|X}\right)+(1-\mu)\mathscr{S}_{\rm{W}}\left(\rho'_{A|X}\right),
\end{equation}
 the normalized assemblage
\begin{eqnarray}
\nonumber\tilde{\rho}^{(\mu)}_{A|X}\coloneqq  \frac{1}{\nu^{(\mu)}}&\bigg[&\mu\, \mathscr{S}_{\rm{W}}\left(\rho_{A|X}\right)\tilde{\rho}_{A|X}\\&+&(1-\mu)\mathscr{S}_{\rm{W}}\left(\rho'_{A|X}\right)\tilde{\rho}'_{A|X}\bigg]
\end{eqnarray}
and the normalized unsteerable assemblage
\begin{eqnarray}
\sigma^{(\mu)}_{A|X}\coloneqq\,\frac{1}{1-\nu^{(\mu)}}&\bigg[&\mu \left(1-\mathscr{S}_{\rm{W}}\left(\rho_{A|X}\right)\right)\sigma_{A|X}\\\nonumber &+&(1-\mu)\left(1-\mathscr{S}_{\rm{W}}\left(\rho_{A|X}\right)\right)\sigma'_{A|X}\bigg].
\end{eqnarray}
Thus, the expression \eqref{eq:second_eq} gives a decomposition of the mixture $\mu\, \rho_{A|X}+(1-\mu)\rho'_{A|X}$ of the form of Eq. \eqref{eq:decompsw}. However, it is not necessarily the optimal one. Hence, we get 
\begin{equation}
\label{eq:conv_W_last}
\mathscr{S}_{\rm{W}}(\mu \rho_{A|X}+(1-\mu)\rho'_{A|X})\leq \nu^{(\mu)},
\end{equation}
 which, together with Eq.  \eqref{eq:def_nu_mu}, finishes the proof of convexity of $\mathscr{S}_{\rm{W}}$. 

Similar arguments can be employed  to prove the theorem's statement concerning the robustness of steering. That $\mathscr{S}_{\rm{Rob}}$ satisfies condition $i$) of Definition \ref{def:Mon_Conv} also follows immediately by definition. Condition $ii$) can be proven with a similar strategy to that for $\mathscr{S}_{\rm{W}}$. Condition $iii$) can be proven by noting that Definition \ref{def:Rob} implies that 
\begin{subequations}
\begin{align}
\rho_{A|X}&=[1+\mathscr{S}_{\rm{Rob}}(\rho_{A|X})]\sigma_{A|X}-\mathscr{S}_{\rm{Rob}}(\rho_{A|X})\tilde{\rho}_{A|X}\\
\rho'_{A|X}&=[1+\mathscr{S}_{\rm{Rob}}(\rho'_{A|X})]\sigma'_{A|X}-\mathscr{S}_{\rm{Rob}}(\rho'_{A|X})\tilde{\rho}'_{A|X},
\end{align}
\end{subequations}
where the unsteerable assemblage $\sigma'_{A|X}$ and the arbitrary assemblage $\tilde{\rho}'_{A|X}$ play respectively the same roles for $\rho'_{A|X}$ to the ones played by $\sigma_{A|X}$ and  $\tilde{\rho}_{A|X}$ for $\rho_{A|X}$ in Definition \ref{def:Rob}.
Then, one can introduce
the positive real 
\begin{equation}
\nonumber\nu^{(\mu)}=\mu\,\mathscr{S}_{\rm{Rob}}(\rho_{A|X}) +(1-\mu)\mathscr{S}_{\rm{Rob}}(\rho'_{A|X}),
\end{equation}
 the normalized assemblage
\begin{equation}
\nonumber\tilde{\rho}^{(\mu)}_{A|X}=\frac{\mu\, \mathscr{S}_{\rm{Rob}}(\rho_{A|X})}{\nu^{(\mu)}}\tilde{\rho}_{A|X}+(1-\mu)\mathscr{S}_{\rm{Rob}}(\rho'_{A|X}) \tilde{\rho}'_{A|X},
\end{equation}
and the normalized unsteerable assemblage
\begin{eqnarray}
\nonumber \sigma^{(\mu)}_{A|X}&=&\frac{\mu \left(1+\mathscr{S}_{\rm{Rob}}(\rho_{A|X})\right)}{1+\nu^{(\mu)}}\sigma_{A|X}\\\nonumber&+&(1-\mu)\left(1+\mathscr{S}_{\rm{Rob}}(\rho'_{A|X})\right) \sigma'_{A|X},
\end{eqnarray}
such that
\begin{eqnarray}
\nonumber \mu\, \rho_{A|X}+(1-\mu)\rho'_{A|X}=\left(1+\nu^{(\mu)}\right)\sigma^{(\mu)}_{A|X}-\nu^{(\mu)} \tilde{\rho}^{(\mu)}_{A|X},
\end{eqnarray}
and proceed with $\mathscr{S}_{\rm{Rob}}$ analogously as with $\mathscr{S}_{\rm{W}}$ in Eq. \eqref{eq:conv_W_last} above.
\end{proof}

%To end up with, it is important to mention that both $\mathscr{S}_{\rm{W}}$ and $\mathscr{S}_{\rm{Rob}}$ satisfy a stronger form of monotonicity than condition $ii$) of Definition \ref{def:Mon_Conv}. It is clear from the proof of Theorem \ref{theo:W_and_Rob} that the two measures do not increase on average not only under deterministic 1W-LOCCs, as in Eq. \eqref{eq:def_strong_mon}, but also under stochastic ones. Indeed, for the steerable weight, Eqs. \eqref{eq:wth_term_asterix} and \eqref{eq:wth_term} make it explicitly evident that the post-selected assemblage $\frac{\mathcal{M}_{\omega}(\rho_{A|X})}{\Tr\left[\mathcal{M}_{\omega}(\rho_{A|X})\right]}$ resulting from the $\omega$-th measurement outcome has itself (even before averaging over $\omega$) smaller or equal steering than $\rho_{A|X}$. The same fact holds for the robustness of steering, as the reader may straightforwardly verify. This  property is inherited from the quantum counterparts for entanglement of these two steering measures, the best-separable approximation entanglement measure \cite{Lewenstein98,Karnas00} and the robustness of entanglement \cite{VidalTarrach99}, both of which also satisfy this highly restrictive form of monotonicity. The relative entropy of steering introduced here is not subject to this restriction. 

%%%%%%%%%%%%%%%%%%%%%%%%%%%%%%%%%%%%%%%%%%%%%%%%%%%%%%%%%%%%%
%%%%%%%%%%%%%%%%%%%%%%%%%%%%%%%%%%%%%%%%%%%%%%%%%%%%%%%%%%%%%
\section{Proof of Theorem \ref{thm:aseemb_tranform}}
\label{sec:assemb_transformations}
\begin{proof}[Proof of Theorem \ref{thm:aseemb_tranform}]
One of the implications is trivial to prove. If $P'_{A|X}=P_{A|X}$ and Eq. \eqref{eq:overlap} holds,  there exists a unitary operator $U$ such that $\ket{\psi'(a,x)}=U\ket{\psi(a,x)}$ for all $a$ and $x$. Then, $\Psi_{A|X}$ can be  transformed into $\Psi'_{A|X}$ by means of a deterministic 1W-LOCC: namely, the one consisting of Bob applying $U$ to his subsystem and Alice doing nothing.
Likewise, if $\Psi'_{A|X}\in\mathsf{LHS}$, then $\Psi_{A|X}$ can trivially be transformed into $\Psi'_{A|X}$ by 1W-LOCCs, as any unsteerable assemblage
can  be created by stochastic 1W-LOCCs by definition (see discussion after Eq. \eqref{eq:lhs}).

Let us then prove the converse implication. That is, assuming that $\Psi_{A|X}$ and $\Psi'_{A|X}$ are pure orthogonal assemblages and that the latter can be obtained from the former by a stochastic 1W-LOCC, we prove that either  $\Psi'_{A|X}\in\mathsf{LHS}$ or $P'_{A|X}=P_{A|X}$ and Eq. \eqref{eq:overlap} is true. To this end, we first note that the no-signaling condition \eqref{eq:nosignaling} restricts minimal-dimension pure orthogonal assemblages to a rather specific form. Namely, the fact that $\Psi_{A|X}$ is no-signaling implies that
\begin{compactenum}[$i$)]
\item either $P_{A|X}(\cdot,x)$ is a deterministic distribution for all $x$,
\item or $P_{A|X}(\cdot,x)$ is the uniform distribution for all $x$.
\end{compactenum}
If case $i$) holds, $\Psi_{A|X} \in \mathsf{LHS}$. Then, since, by assumption, $\Psi'_{A|X}$ can be obtained via a stochastic 1W-LOCC from $\Psi_{A|X}$, one automatically obtains that $\Psi'_{A|X} \in \mathsf{LHS}$. 

To analyze case  $ii$), we use that $\Psi'_{A|X}$ is also subject to the no-signaling condition \eqref{eq:nosignaling}:
\begin{compactenum}[$i'$)]
\item either $P'_{A|X}(\cdot,x)$ is a deterministic distribution for all $x$, 
\item or $P'_{A|X}(\cdot,x)$ is the uniform distribution for all $x$.
\end{compactenum}
That case $i'$) is possible if case  $ii$) holds is clear, as $i'$) corresponds to $\Psi'_{A|X} \in \mathsf{LHS}$. So, it only remains to show that if cases $ii$) and $ii'$) hold, then either $\Psi'_{A|X}\in\mathsf{LHS}$ or Eq. \eqref{eq:overlap} holds. We show it in what follows.

Assuming that $ii$) and $ii'$) hold and that there is a stochastic 1W-LOCC $\mathcal{M}$ that maps $\Psi_{A|X}$ into  $\Psi'_{A|X}$, i.e., such that $\mathcal{M}(\Psi_{A|X})\propto\Psi'_{A|X}$, where ``$\propto$" stands for ``is proportional to", we use Eq.~\eqref{eq:prot6} to obtain
\begin{widetext}
\begin{equation}
\label{eq:protpure0}
\sum_{a,x,\omega}P_{X|X_f,\Omega}(x,x_f,\omega)P_{A_f|A,X,\Omega, X_f}(a_f,a,x,\omega,x_f) K_{\omega} \proj{\psi(a,x)} K^{\dagger}_{\omega}\propto\proj{\psi'(a_f,x_f)}\ \forall\  (a_f, x_f).
\end{equation}
Since the right-hand side of Eq.~\eqref{eq:protpure0} is composed of a  rank-one projector onto a pure state, each term of the sum in the left-hand side must be either zero or proportional to $\proj{\psi'(a_f,x_f)}$. In particular, this must also hold for each $\omega$-th term. That is, for all $\omega$, it must hold that 
\begin{equation}
\label{eq:protpure_no_omega}
\sum_{a,x}P_{X|X_f,\Omega}(x,x_f,\omega)P_{A_f|A,X,\Omega, X_f}(a_f,a,x,\omega,x_f) K_{\omega} \proj{\psi(a,x)} K^{\dagger}_{\omega}\sim\proj{\psi'(a_f,x_f)}\ \forall\  (a_f, x_f),
\end{equation}
 \end{widetext}
where the symbol ``$\sim$" is used to signify ``is either equal to zero or proportional to".
 Indeed, using that $K_{\omega}\neq0$ and that $P_{X|X_f,\Omega}$ and $P_{A_f|A,X,\Omega, X_f}$ are normalised distributions, one can see by case analysis that there are always at least two different pairs $(a_f,x_f)$ for which the  left-hand side of Eq. \eqref{eq:protpure_no_omega} is not zero and, therefore, proportional to $\proj{\psi'(a_f,x_f)}$.

Let us then consider first the case 
\begin{equation}
\label{eq:K_omega_fijo}
K_{\omega} \ket{\psi(a,x)}\neq0 \ \forall\  (a, x).
\end{equation}
The other case will be considered at the end. The first step is to note that Eqs. \eqref{eq:protpure_no_omega} and \eqref{eq:K_omega_fijo} imply that, unless $\Psi'_{A|X} \in \mathsf{LHS}$, 
 \begin{subequations}
\label{eq:P_delta}
\begin{align}
\label{eq:P_delta_x}
P_{X|X_f,\omega}(x|x_f,\omega)&=\delta_{x_f\,  x \oplus f(\omega)},\\
\label{eq:P_delta_a}
P_{A_f|A,X,\Omega, X_f}(a_f,a,x,\omega,x \oplus f(\omega))&\in\{0,1\} \ \forall\  (a_f,a,x),
\end{align}
\end{subequations}
where $f(\omega)\in\{0, 1\}$. That is, for any  $\omega$ for which Eq. \eqref{eq:K_omega_fijo} holds, unless $\Psi'_{A|X} \in \mathsf{LHS}$, the variables $X$ and $X_f$ must be either fully correlated or fully anticorrelated and $P_{A_f|A,X,\Omega, X_f}(\cdot,a_f,x,\omega,x \oplus f(\omega))$ must be a deterministic distribution for all $(a_f,x)$. 

To prove Eq. \eqref{eq:P_delta_x}, suppose that it does not hold. Then there must exist $\tilde{x}$ such that $P_{X|X_f,\omega}(\tilde{x}|x_f,\omega)\neq 0$ for all $x_f$. This, due to Eq. \eqref{eq:protpure_no_omega}, implies that 
\begin{subequations}
\label{eq:protpure_no_x}
\begin{align}
&\nonumber\sum_{a}P_{A_f|A,X,\Omega, X_f}(a_f,a,\tilde{x},\omega,0) K_{\omega} \proj{\psi(a,\tilde{x})} K^{\dagger}_{\omega} \\
&\hspace{1cm}\sim\proj{\psi'(a_f,0)},\\
\nonumber &\sum_{a}P_{A_f|A,X,\Omega, X_f}(a_f,a,\tilde{x},\omega,1) K_{\omega} \proj{\psi(a,\tilde{x})} K^{\dagger}_{\omega}\\
&\hspace{1cm}\sim\proj{\psi'(a_f,1)}.
\end{align}
\end{subequations}
In turn, choosing $\tilde{a}_f$ and $\overline{a}_f$ such that $P_{A_f|A,X,\Omega, X_f}(\tilde{a}_f,a,\tilde{x},\omega,0)>0$ and $P_{A_f|A,X,\Omega, X_f}(\overline{a}_f,a,\tilde{x},\omega,1)>0$, which is always possible due to $P_{A_f|A,X,\Omega, X_f}$ being a normalised distribution and does not require any extra assumption, Eqs. \eqref{eq:K_omega_fijo} and \eqref{eq:protpure_no_x} imply that
\begin{subequations}
\label{eq:protpure_no_a}
\begin{align}
K_{\omega} \ket{\psi(a,\tilde{x})} 
&\propto\ket{\psi'(\tilde{a}_f,0)},\\
K_{\omega} \ket{\psi(a,\tilde{x})} 
&\propto\ket{\psi'(\overline{a}_f,1)}.
\end{align}
\end{subequations}
This finally leads to $\ket{\psi'(\tilde{a}_f,0)}=\ket{\psi'(\overline{a}_f,1)}$, which is true only if $\Psi'_{A|X} \in \mathsf{LHS}$. 

To prove Eq. \eqref{eq:P_delta_a} we use a similar argument. If one assumes that Eq. \eqref{eq:P_delta_a} is false, then there must exist a pair $(\tilde{a}, \tilde{x})$ such that $P_{A_f|A,X,\Omega, X_f}(a_f,\tilde{a},\tilde{x},\omega,\tilde{x} \oplus f(\omega))> 0$ for all $a_f$. Using this and Eqs. \eqref{eq:protpure_no_omega}, \eqref{eq:K_omega_fijo}, and \eqref{eq:P_delta_x}, one arrives at 
\begin{subequations}
\label{eq:protpure_final}
\begin{align}
K_{\omega} \ket{\psi(\tilde{a},\tilde{x})}  
&\propto\ket{\psi'(0,\tilde{x}\oplus f(\omega))},\\
K_{\omega} \ket{\psi(\tilde{a},\tilde{x})} 
&\propto\ket{\psi'(1,\tilde{x}\oplus f(\omega))},
\end{align}
\end{subequations}
which, since $\ket{\psi'(0,\tilde{x}\oplus f(\omega))}$ and $\ket{\psi'(1,\tilde{x}\oplus f(\omega))}$ are orthogonal, yields a contradiction.

The second step is to note that Eqs. \eqref{eq:K_omega_fijo} and \eqref{eq:P_delta} impose restrictions on which $a$'s and $x$'s can contribute to each $a_f$ and $x_f$ in Eq.~\eqref{eq:protpure_no_omega}. More precisely, one can see by case analysis that, up to relabelings of $a_f$ or $x_f$, only three different types of assignments are possible: 
\begin{widetext}
\begin{center}
\begin{tabular}{ l | l | l }
  a) & b) & c)\\ \hline
  $K_{\omega}\ket{\psi(0,0)} \propto \ket{\psi'(0,0)}$ & $K_{\omega}\ket{\psi(0,0)} \propto  \ket{\psi'(0,0)}$ & $K_{\omega}\ket{\psi(0,0)} \propto  \ket{\psi'(0,0)}$ \\
  $K_{\omega}\ket{\psi(1,0) }\propto  \ket{\psi'(1,0)}$ & $K_{\omega}\ket{\psi(1,0) } \propto  \ket{\psi'(1,0)}$ & $K_{\omega}\ket{\psi(1,0) } \propto  \ket{\psi'(0,0)}$  \\
   $K_{\omega
   }\ket{\psi(0,1)} \propto \ket{\psi'(0,1)}$ & $K_{\omega}\ket{\psi(0,1)}  \propto\ket{\psi'(0,1)}$ & $K_{\omega}\ket{\psi(0,1)}  \propto \ket{\psi'(0,1)}$ \\
    $K_{\omega
    }\ket{\psi(1,1) }\propto \ket{\psi'(1,1)}$ & $K_{\omega}\ket{\psi(1,1) } \propto\ket{\psi'(0,1)}$ & $K_{\omega}\ket{\psi(1,1) }\propto \ket{\psi'(0,1)}$\\
\end{tabular}
\end{center}
\end{widetext}

The third step is to show that all three cases a-c) are possible only if either $\Psi'_{A|X} \in \mathsf{LHS}$ or Eq. \eqref{eq:overlap} holds. Note that it is enough to show this for the case where all the eight vectors $\{\ket{\psi(a,x)},\ket{\psi'(a_f,x_f)}\}_{a,x,a_f,x_f}$ lie on a same plane of the Bloch sphere. This is due to the fact that, since $\Psi_{A|X}$ and $\Psi'_{A|X}$ are both pure no-signaling assemblages of minimal dimension, $\{\ket{\psi(a,x)}\}_{a_f,x_f}$ and $\{\ket{\psi'(a_f,x_f)}\}_{a_f,x_f}$ are each one already contained in two planes of the Bloch sphere, as one can straightforwardly see using Eq. \eqref{eq:nosignaling}. These two planes can always be rotated so as to coincide by a unitary operation, which can in turn be absorbed in the definition of the Kraus operator $K_{\omega}$. Hence, without loss of generality, we take 
\begin{subequations}
\label{eq:states_same_plane}
\begin{align*}
\ket{\psi(0,0)}&=\ket{0},\\ \ket{\psi(1,0)}&=\ket{1},\\
\ket{\psi(0,1)}&=\cos(\varphi)\ket{0}+
\sin(\varphi)\ket{1},\\ \ket{\psi(1,1)}&=-\sin(\varphi)\ket{0}+\cos(\varphi)\ket{1},\\
\ket{\psi'(0,0)}&=\cos(\theta)\ket{0}+
\sin(\theta)\ket{1},\\ \ket{\psi'(1,0)}&=-\sin(\theta)\ket{0}+\cos(\theta)\ket{1},\\
\ket{\psi'(1,1)}&=\cos(\phi)\ket{0}+
\sin(\phi)\ket{1},\\ \ket{\psi'(1,1)}&=-\sin(\phi)\ket{0}+\cos(\phi)\ket{1},
\end{align*}
\end{subequations}
for  arbitrary  $\varphi$, $\theta$ and $\phi\in[0,\pi/2[$, and where $\ket{0}$ and $\ket{1}$ represent the computational-basis states. We analyse first the case a). Dividing both vector components (in the computational basis) of the first equation of this case, one obtains that $\frac{\left[K_{\omega}\right]_{00}}{\left[K_{\omega}\right]_{10}}=\frac{\cos{(\theta)}}{\sin{(\theta)}}$, where $\left[K_{\omega}\right]_{ij}\coloneqq\bra{i}K_{\omega}\ket{j}$. Analogously, dividing both vector components of the second equation yields $\frac{\left[K_{\omega}\right]_{01}}{\left[K_{\omega}\right]_{11}}=\frac{-\sin{(\theta)}}{\cos{(\theta)}}$. Hence, introducing proportionality constants $\kappa_1>0$ and $\kappa_2>0$, the Kraus operator can be matrix-represented in the computational basis as
\begin{equation}
\label{eq:K_omega_kappa}
K_{\omega}=\begin{pmatrix}
\kappa_1\cos(\theta) & -\kappa_2\sin(\theta) \\
\kappa_1\sin(\theta) & \kappa_2\cos(\theta)
\end{pmatrix}.
\end{equation} 
Using Eq. \eqref{eq:K_omega_kappa}, the third equation of case a) implies  that 
\begin{equation}\label{eq:trig1}
\kappa_1 \cos(\varphi) \sin(\theta-\phi)=\kappa_2\sin(\varphi)\cos(\theta- \phi).
\end{equation}
Finally, the fourth equation leads to
\begin{equation}\label{eq:trig2}
\kappa_2 \cos(\varphi) \sin(\theta-\phi)=\kappa_1\sin(\varphi)\cos(\theta- \phi).
\end{equation}
Eqs. \eqref{eq:trig1} and \eqref{eq:trig2} can be simultaneously satisfied only if $\theta-\phi=\varphi$ or $(\theta-\phi)\times\varphi=0$. The former option yields Eq. \eqref{eq:overlap}. The latter one implies that $\Psi'_{A|X} \in \mathsf{LHS}$.
In a similar fashion, for case b), the first three equations lead to Eq. \eqref{eq:trig1} and the fourth one  to
\begin{equation}
-\kappa_1\sin(\varphi)\sin(\theta-\phi)=\kappa_2\cos(\phi)\cos(\theta- \phi).
\end{equation}
This cannot be satisfied unless  $(\theta-\phi)\times\varphi=0$, which means that $\Psi'_{A|X} \in \mathsf{LHS}$. 
With a similar argument the reader can straightforwardly verify that the same thing happens for case c). This finishes the proof of the theorem for the $\omega$'s for which Eq. \eqref{eq:K_omega_fijo} holds.

As the fourth and final step, it remains to treat the case where, for a certain $\omega$, there exists a pair $(\tilde{a},\tilde{x})$ for which $K_{\omega} \ket{\psi(\tilde{a},\tilde{x})}=0$. Since $K_{\omega}\neq0$, the latter is true only if the support of $K_{\omega}$ is given by the span of $\ket{\psi(\tilde{a}\oplus1,\tilde{x})}$. 
%Furthermore, since $P_{X|X_f,\Omega}$ and $P_{A_f|A,X,\Omega, X_f}$ are normalised distributions, one can see that there is always a pair  $(\tilde{a}_f,\tilde{x}_f)$ for which $P_{X|X_f,\Omega}(\tilde{x},\tilde{x}_f,\omega)P_{A_f|A,X,\Omega, X_f}(\tilde{a}_f,\tilde{a}\oplus1,\tilde{x},\omega,\tilde{x}_f)>0$. This, together with Eq. \eqref{eq:protpure_no_omega}, leads to $K_{\omega}\propto\ket{\psi'(\tilde{a}_f,\tilde{x}_f)}\bra{\psi(\tilde{a}\oplus1,\tilde{x})}$. On the other hand, unless  $\Psi'_{A|X} \in \mathsf{LHS}$, one has in addition that $K_{\omega} \ket{\psi(a,\tilde{x}\oplus1)}\neq0$ for all $a$. Hence, assuming that $\Psi'_{A|X} \notin \mathsf{LHS}$, it holds that
%\begin{equation}
%\label{eq:K_not_zero}
%K_{\omega} \ket{\psi(a,x)}\propto
%\begin{cases}
%    \ket{\psi'(\tilde{a}_f,\tilde{x}_f)}, & \text{if}\ (a,x)\neq(\tilde{a},\tilde{x}),\\
%    0, & \text{if}\ (a,x)=(\tilde{a},\tilde{x}).
%  \end{cases}
%\end{equation}
Using this and the fact that there are always at least two different pairs $(a_f,x_f)$ for which the left-hand side of Eq. \eqref{eq:protpure_no_omega} is not zero, one obtains that  %$\ket{\psi'(\tilde{a}_f,\tilde{x}_f)}\propto\ket{\psi'(a_f,x_f)}$ 
$K_{\omega}\propto\ket{\psi'(a_f,x_f)}\bra{\psi(\tilde{a}\oplus1,\tilde{x})}$ for two different pairs $(a_f,x_f)$, which, unless $\Psi'_{A|X}\in \mathsf{LHS}$, is a contradiction.
\end{proof}

%%%%%%%%%%%%%%%%%%%%%%%%%%%%%%%%%%%%%%%%%%%%%%%%%%%%%%%%%%%%%
%%%%%%%%%%%%%%%%%%%%%%%%%%%%%%%%%%%%%%%%%%%%%%%%%%%%%%%%%%%%%
\section{Non-existence of minimal-dimension steering bits}
\label{sec:no_steering_bits}
In this appendix we prove Theorem \ref{thm:nosteeringbig}. This section bears many similarities with App. \ref{sec:assemb_transformations}.
 
\begin{proof}[Proof of Theorem \ref{thm:nosteeringbig}] 
We proceed by {\it reductio ad absurdum}. That is, we show that if one supposes that there exists a pure normalised assemblage $\Psi_{A|X}\coloneqq\{P_{A|X}(a,x),\ket{\psi(a,x)}\}_{a,x}$, with $d=s=r=2$, from which all assemblages can be obtained via stochastic 1W-LOCCs, one obtains a contradiction. %As a matter of fact, we show that it suffices to assume only that  $\Psi_{A|X}$ can be mapped into two steerable pure orthogonal assemblages of minimal dimension of a very specific form to arrive at a contradiction.

%First, we recall that, as mentioned in App. \ref{sec:assemb_transformations}, since $\Psi_{A|X}$ is a pure no-signaling assemblage of minimal dimension, the four state vectors $\{\ket{\psi(a,x)}\}_{a,x}$ are all contained in a single plane of the Bloch sphere. 
Without loss of generality, we can choose the computational basis $\{\ket{0},\ket{1}\}$ so that its first element coincides with $\ket{\psi(0,0)}$ and the element $\ket{\psi(1,0)}$ is in the plane that contains the vectors  $\ket{0}$ and $\frac{1}{\sqrt{2}}(\ket{0}+\ket{1})$. What is more, clearly, $\Psi_{A|X}$ cannot have a LHS model, otherwise $\Psi_{A|X}$ could not be mapped into all assemblages by stochastic 1W-LOCCs. Thus, we can safely assume that 
\begin{equation}
\label{eq:Psi_no_LHS}
\Psi_{A|X}\notin\mathsf{LHS}.
\end{equation}
 Hence, we take 
\begin{equation}
\label{eq:psi_varphi}
\begin{split}
&\ket{\psi(0,0)}=\ket{0}\\
&\ket{\psi(1,0)}= \cos(\varphi_{10})\ket{0}+\sin(\varphi_{10})\ket{1} \\
&\ket{\psi(0,1)}= \cos(\varphi_{01})\ket{0}+e^{i\alpha_{01}}\sin(\varphi_{01})\ket{1}  \\
&\ket{\psi(1,1)}= \cos(\varphi_{11})\ket{0}+e^{i\alpha_{11}}\sin(\varphi_{11})\ket{1}  
\end{split}
\end{equation}
with 
\begin{subequations}
\label{eq:def_angles}
\begin{align}
&\varphi_{10}\in\ ]0,\pi[\\
&\alpha_{ax} \in\ [0,2\pi],\ \forall\ (a,x) \notin \{(0,0),(1,0)\}
\end{align}
\end{subequations}
and
\begin{align}
\label{eq:cond_angles}
(\varphi_{a1},\alpha_{a,1}) \neq(\varphi_{a'1},\alpha_{a'1})\ \forall\  a\neq a'.
\end{align}
Equations \eqref{eq:def_angles} and \eqref{eq:cond_angles} hold due to the fact that $\Psi_{A|X}\notin\mathsf{LHS}$ and the no-signaling condition \eqref{eq:nosignaling}. More precisely, if $\varphi_{10}=\{0,\pi\}$, $\ket{\psi(1,0)}=\ket{0}$, which implies that Bob's reduced state is $\varrho_B=\proj{0}$. Then, the no-signalling condition \eqref{eq:nosignaling} implies that $\ket{\psi(0,1)}=\ket{0}=\ket{\psi(1,1)}$. Such assemblage clearly has a LHS model, which contradicts the assumption \eqref{eq:Psi_no_LHS}. The same argument implies \eqref{eq:cond_angles}.
Furthermore, $\Psi_{A|X}\notin\mathsf{LHS}$ and the no-signaling principle imply also that $P_{A|X}(a,x)\neq 0$ for all $(a,x)$. To see the latter, suppose that there is a pair $(a,x)$ for which $P_{A|X}(a,x)= 0$. Then, clearly, $P_{A|X}(a\oplus1,x)=1$. This, together with Eq. \eqref{eq:nosignaling}, implies that there is an $\tilde{a}$ for which $P_{A|X}(\tilde{a},x \oplus 1)=1$, which in turn leads to $\Psi_{A|X}\in\mathsf{LHS}$. 

Let us now consider pure orthogonal assemblages $\{\Psi_{A|X}^{\theta}\}_{\theta}$ with $d=s=r=2$ of the form $\Psi_{A|X}^{\theta}\coloneqq\{\frac{1}{2},\ket{\psi^{\theta}(a,x)}\}_{a,x}$, where
\begin{subequations}
\label{eq:psi_theta}
\begin{equation}
\label{eq:psi_theta_a}
\begin{split}
&\ket{\psi^{\theta}(0,0)} =\ket{0},\\ & \ket{\psi^{\theta}(1,0)}=\ket{1}
\end{split}
\end{equation}
\text{and}
\begin{equation}
\label{eq:psi_theta_b}
\begin{split}
&\ket{\psi^{\theta}(0,1)}=\cos(\theta)\ket{0}+
\sin(\theta)\ket{1},\\& \ket{\psi^{\theta}(1,1)}=-\sin(\theta)\ket{0}+\cos(\theta)\ket{1}.
\end{split}
\end{equation}
\end{subequations} 
We restrict to $0<\theta<\pi/2$ to ensure that $\Psi_{A|X}^{\theta}\notin\mathsf{LHS}$. If all assemblages can be obtained via stochastic 1W-LOCCs from $\Psi_{A|X}$, there must be a stochastic 1W-LOCC $\mathcal{M}^{\theta}$ such $\mathcal{M}^{\theta}(\Psi_{A|X})\propto\Psi^{\theta}_{A|X}$, where ``$\propto$" stands for ``is proportional to". Then, as in App. \ref{sec:assemb_transformations}, Eq.~\eqref{eq:prot6} implies that, for all $\omega$, it must hold that 
\begin{align}
\label{eq:M_theta_omega}
\nonumber
\sum_{a,x}&P^{\theta}_{X|X_f,\Omega}(x,x_f,\omega)P^{\theta}_{A_f|A,X,\Omega, X_f}(a_f,a,x,\omega,x_f)\\\nonumber &\times P_{A|X}(a,x) K^{\theta}_{\omega} \proj{\psi(a,x)} {K^\theta_{\omega}}^{\dagger}\\
&\hspace{1cm}\sim\proj{\psi^{\theta}(a_f,x_f)}\ \ \ \ \forall\ (a_f,x_f),
\end{align}
where the symbol ``$\sim$" is used to signify ``is either equal to zero or proportional to". However, we note again that, since $K_{\omega}\neq0$ and $P_{X|X_f,\Omega}$ and $P_{A_f|A,X,\Omega, X_f}$ are normalised distributions, there are always at least two different pairs $(a_f,x_f)$ for which the left-hand side of Eq.~\eqref{eq:M_theta_omega} is not zero and, therefore, proportional to $\proj{\psi^{\theta}(a_f,x_f)}$, as can be seen by direct case analysis.

Let us then  consider the case 
\begin{equation}
\label{eq:K_theta_omega_fijo}
K^{\theta}_{\omega} \ket{\psi(a,x)}\neq0 \ \forall\ (a,x).
\end{equation}
The other case will be considered later. The first step is to note that Eqs. \eqref{eq:M_theta_omega} and \eqref{eq:K_theta_omega_fijo}, together with the fact that $\Psi_{A|X}\notin\mathsf{LHS}$, imply that
 \begin{subequations}
\label{eq:P_delta_theta}
\begin{align}
\label{eq:P_delta_x_theta}
P^{\theta}_{X|X_f,\omega}(x|x_f,\omega)&=\delta_{x_f\, x \oplus f^{\theta}(\omega)},\\
\label{eq:P_delta_a_theta}
P^{\theta}_{A_f|A,X,\Omega, X_f}(a_f,a,x,\omega,x \oplus f^{\theta}(\omega))&\in\{0,1\}  \ \forall\  (a_f,a,x),
\end{align}
\end{subequations}
where $f^{\theta}(\omega)\in\{0, 1\}$. That is, for any  $\omega$ for which Eq. \eqref{eq:K_theta_omega_fijo} holds, $X$ and $X_f$ must be either fully correlated or fully anticorrelated and $P^{\theta}_{A_f|A,X,\Omega, X_f}(\cdot,a_f,x,\omega,x \oplus f(\omega))$ must be a deterministic distribution for all $(a_f,x)$. The proofs of Eqs.~\eqref{eq:P_delta_theta} are almost identical to the proofs of Eqs.~\eqref{eq:P_delta} in App. \ref{sec:assemb_transformations}, with the only difference that, here, $\Psi_{A|X}\notin\mathsf{LHS}$ and $\Psi^{\theta}_{A|X}\notin\mathsf{LHS}$ are true by assumption. We therefore do not repeat the argument.

The second step is to note that Eqs. \eqref{eq:K_theta_omega_fijo} and \eqref{eq:P_delta_theta} impose restrictions on which $a$'s and $x$'s can contribute to each $a_f$ and $x_f$ in Eq.~\eqref{eq:M_theta_omega}. More precisely, one can see by case analyses that, up to relabelings of $a_f$ or $x_f$, only one type of assignment is possible: 
\begin{subequations}
\label{eq:assignment}
\begin{align}
\label{eq:setcond1}  K^{\theta}_{\omega}\ket{\psi(0,0)} &\propto \ket{\psi^{\theta}(0,0)},\\
\label{eq:setcond2}   K^{\theta}_{\omega}\ket{\psi(1,0)} &\propto  \ket{\psi^{\theta}(1,0)},\\
\label{eq:setcond3}   K^{\theta}_{\omega}\ket{\psi(0,1)} &\propto  \ket{\psi^{\theta}(0,1)},\\
\label{eq:setcond4}   K^{\theta}_{\omega}\ket{\psi(1,1)} &\propto  \ket{\psi^{\theta}(1,1)}.
\end{align}
\end{subequations}

The third step is to show that Eqs.~\eqref{eq:assignment} lead to a contradiction. To this end, together with Eqs.~\eqref{eq:psi_varphi} and \eqref{eq:psi_theta_a}, Eqs. \eqref{eq:setcond1} and \eqref{eq:setcond2} respectively imply that $\left[K^{\theta}_{\omega}\right]_{10}=0$ and $\frac{\left[K^{\theta}_{\omega}\right]_{00}}{\left[K^{\theta}_{\omega}\right]_{01}}=-\tan\left(\varphi_{10}\right)$, where $\left[K^{\theta}_{\omega}\right]_{ij}\coloneqq \bra{i}K^{\theta}_{\omega}\ket{j}$. In turn, dividing both vector components in each one of Eqs. \eqref{eq:setcond3} and \eqref{eq:setcond4}, one obtains, using Eqs.~\eqref{eq:psi_varphi} and \eqref{eq:psi_theta_b}, that 
\begin{subequations}
\begin{align}
\label{eq:pruebacond1}\frac{[K_{\omega}^{\theta}]_{11} }{[K_{\omega}^{\theta}]_{01} }&=\tan(\theta)\left(- \frac{\tan(\varphi_{10})}{\tan(\varphi_{01})e^{i\alpha_{01}}}+1 \right),\\
\label{eq:pruebacond2}\frac{[K_{\omega}^{\theta}]_{11} }{[K_{\omega}^{\theta}]_{01} }&=\frac{-1}{\tan(\theta)}\left(-\frac{\tan(\varphi_{10})}{\tan(\varphi_{11})e^{i\alpha_{11}}}+1 \right).
\end{align}
\end{subequations}
Equating the right-hand sides of Eqs. \eqref{eq:pruebacond1} and \eqref{eq:pruebacond2} gives, after straightforward algebraic manipulation, 
\begin{equation}
\label{eq:conditiontrig}
\frac{1}{\tan(\varphi_{10})}=\frac{\sin^2(\theta)}{\tan(\varphi_{11})e^{i\alpha_{11}}}+\frac{\cos^2 \theta}{\tan(\varphi_{01})e^{i\alpha_{01}}}.
\end{equation}
Since the last condition is independent of $K_{\omega}^{\theta}$ and $\Psi_{A|X}$ should be transformed by stochastic 1W-LOCCs into any member of the family $\{\Psi_{A|X}^{\theta}\}_{\theta}$, the same condition should be fulfilled for any $0<\theta<\pi/2$. It actually suffices to choose just two assemblages $\Psi_{A|X}^{\theta_1}$ and $\Psi_{A|X}^{\theta_2}$, for any $0<\theta_1,\ \theta_2<\pi/2$ with $\theta_1\neq\theta_2$, to arrive at a contradiction. Indeed, since the angles $\varphi_{10}$, $\varphi_{01}$, $\varphi_{11}$, $\alpha_{01}$ and $\alpha_{11}$ are fixed, the only way to satisfy Eq. \eqref{eq:conditiontrig} for both $\theta_1$ and $\theta_2$ is that
\begin{equation}
\tan(\varphi_{10})=\tan(\varphi_{01})e^{i\alpha_{01}}=\tan(\varphi_{11})e^{i\alpha_{11}}.
\end{equation}
This, in turn, can happen only if $\alpha_{01}=0=\alpha_{11}$ and $\varphi_{10}=\varphi_{01}=\varphi_{11}$, which is clearly incompatible with \eqref{eq:cond_angles}. 

It remains to treat the case where, for a certain $\omega$, there exists a pair $(\tilde{a},\tilde{x})$ for which Eq. \eqref{eq:K_theta_omega_fijo} does not hold. By relabeling $a$ or $x$, we can always choose $(\tilde{a},\tilde{x})=(0,0)$. Hence, we consider 
\begin{equation}
\label{eq:K_theta_omega_fijo_zero}
K^{\theta}_{\omega} \ket{\psi(0,0)}=K^{\theta}_{\omega} \ket{0}=0.
\end{equation}
Since $K^{\theta}_{\omega}\neq0$, the latter is true only if the support of $K^{\theta}_{\omega}$ is given by the span of $\ket{1}$. 
Using this and the fact that there are always at least two different pairs $(a_f,x_f)$ for which the left-hand side of Eq. \eqref{eq:M_theta_omega} is not zero, one arrives at a contradiction of the type  $K^{\theta}_{\omega}\propto\ket{\psi^{\theta}(a_f,x_f)}\bra{1}$ for two different pairs $(a_f,x_f)$. This finishes the proof for pure assemblages.
\end{proof}

We finish the appendix with a remark on a difficulty to generalise Theorem \ref{thm:nosteeringbig} to the case of mixed-state assemblages, i.e., to rule out the existence of steering bits also among mixed-state assemblages. Since any mixed-state assemblage can be decomposed as a convex combination of pure assemblages and $\mathcal{M}$ is a linear transformation, one would be tempted to trivially extend the proof above to mixed-state assemblages by using similar reasonings to those presented just above with each pure assemblage in the convex combination together with linearity arguments. However, such straightforward extension unfortunately fails. The reason for this is that each pure assemblage in the pure-assemblage decomposition of a mixed-state assemblage is, as far as we can see, not necessarily no-signalling. We emphasise that all our formalism deals only with no-signalling objects. Hence, while we strongly believe that minimal-dimension steering bits do not exist in general, i.e., even among the mixed-state assemblages, we leave the proof of this statement as an open question.
%%%%%%%%%%%%%%%%%%%%%%%%%%%%%%%%%%%%%%%%%%%%%%%%%%%%%%%%%%%%%
%%%%%%%%%%%%%%%%%%%%%%%%%%%%%%%%%%%%%%%%%%%%%%%%%%%%%%%%%%%%%
\section{Proofs of Lemmas \ref{lem:strongmonotonicitydistance} and \ref{lem:convex}}
\label{Sec:Proofs_of_Lemmas}
Before we proceed, we recall some known mathematical facts necessary for the proofs.

First, The von-Neumman relative entropy $S_\text{Q}$, defined by Eq. \eqref{eq:von_Neumman_rel_ent}, fulfils the following properties \cite{Lieb73} .
\begin{itemize} 
\item Given two sets $\{\varrho^{(j)}\}_{j=1,\ldots,n}$ and $\{\varrho'^{(j)}\}_{j=1,\ldots,n}$ of $n$ arbitrary positive-semidefinite (not necessarily normalized) operators each and $n$ positive real numbers $\{\mu^{(j)}\}_{j=1,\ldots,n}$ such that $\sum_{j}\mu^{(j)}=1$, with $n\in\mathbb{N}$, $S_\text{Q}$ satisfies the \emph{joint convexity} property 
\begin{equation}
\label{eq:convexity_Sq}
S_\text{Q}\left(\sum_{i}\mu^{(j)} \varrho^{(j)} \Big\|\sum_{j}\mu^{(j)} \varrho'^{(j)}\right) \leq \sum_{j} \mu^{(j)} S_\text{Q}\left(\varrho^{(j)}\big\|\varrho'^{(j)}\right).
\end{equation}
\item Given any  completely-positive  trace-preserving (CPTP) map $\mathcal{E}$ and any two density operators $\varrho$ and $\varrho$, $S_\text{Q}$ satisfies the \emph{CPTP-map contraction} property
\begin{equation}\label{eq:processingineq}
S_\text{Q}\left(\mathcal{E}(\varrho)|\mathcal{E}(\varrho')\right)\leq S_\text{Q}(\varrho|\varrho').
\end{equation}
\end{itemize}

Second, the Kullback-Leibler divergence  $S_\text{C}$ defined in Eq. \eqref{eq:Kullback_Leibler_div} fulfils the following  property. 
\begin{widetext}
\begin{itemize}
\item Given any two joint probability distributions $P_{X,Y}$ and $P_{X,Y}'$ over classical bits $x$ and $y$, $S_\text{C}$ satisfies the inequality
\begin{align}
\nonumber
\sum_{x}P_{X}(x) S_\text{C}\left(P_{Y|X}(\cdot,x)\|P'_{Y|X}(\cdot,x)\right)&=\sum_{x}P_{X}(x)\sum_y P_{Y|X}(x,y) \frac{\log P_{Y|X}(x,y)}{\log P'_{Y|X}(x,y)}\\
\nonumber
&=\sum_{x,y} P_{X,Y}(x,y)\left(\frac{\log P_{X,Y}(x,y)}{\log P'_{X,Y}(x,y)}-\frac{\log P_{X}(x)}{\log P'_{X}(x)}\right)\\
\nonumber
&= S_\text{C}\left(P_{X,Y}\|P'_{X,Y}\right)-S_\text{C}\left(P_{X}\|P'_{X}\right)\\
\label{eq:knullback}
&\leq S_\text{C}\left(P_{X,Y}\|P'_{X,Y}\right).
\end{align}
\end{itemize}

We are now in a good position to prove the lemmas. 
%%%%%%%%%%%%%%%%%%%%%%%%%%%%%%%%%%%%%%%%%%%%%%%%%%%%%%%%%%%%%
\subsection{Proof of Lemma \ref{lem:strongmonotonicitydistance}}
We begin by Lemma \ref{lem:strongmonotonicitydistance}.
\begin{proof}[Proof of Lemma \ref{lem:strongmonotonicitydistance}]
First, using the definition of $\mathscr{S}_{\text{R}}$ in Eq. \eqref{eqmat:relativeentropy}, we write the left-hand side of Eq. \eqref{eq:strongmonotonicitydist} explicitly as
\begin{align}
\nonumber
\sum_{\omega} P_{\Omega}(\omega) S_\text{A}\left( \frac{\mathcal{M}_{\omega}\left(\hat{\rho}_{A|X}\right)}{\Tr\left[\mathcal{M}_{\omega}(\rho_{A|X})\right]}\Bigg\|\frac{\mathcal{M}_{\omega}\left(\hat{\rho}'_{A|X}\right)}{\Tr\left[\mathcal{M}_{\omega}(\rho'_{A|X})\right]}\right)
=
\sum_{\omega}P_{\Omega}(\omega) \max_{P_{X_f|\Gamma},\{}E_{\gamma}\} \Bigg[S_\text{C}\left(P_{\Gamma|\Omega}(\cdot, \omega)\|P'_{\Gamma|\Omega}(\cdot, \omega)\right)+&\\
\label{eq:distancepormu0}
 \sum_{\gamma,x_f} P_{X_f|\Gamma}(x_f,\gamma)\, P_{\Gamma|\Omega}(\gamma,\omega)\, S_\text{Q}\left( \frac{\id\otimes E_{\gamma}\left[\mathcal{M}_{\omega}\left(\hat{\rho}_{A|X}\right)\right](x_f)\id\otimes E_{\gamma}^{\dagger}}
 {P_{\Gamma,\Omega}(\gamma,\omega)}
 \:\bigg\|\:\frac{\id\otimes E_{\gamma} \left[\mathcal{M}_{\omega}\left(\hat{\rho}'_{A|X}\right)\right](x_f)\id\otimes E_{\gamma}^{ \dagger}}
 {P'_{\Gamma,\Omega}(\gamma,\omega)} \right)\Bigg],
 \end{align}
where we have used that $P_{\Omega}(\omega)=\Tr[\mathcal{M}_{\omega}(\hat{\rho}_{A|X})]$ and $P'_{\Omega}(\omega)=\Tr[\mathcal{M}_{\omega}(\hat{\rho}'_{A|X})]$, and that  $P_{\Gamma,\Omega}(\gamma,\omega)=P_{\Gamma|\Omega}(\gamma,\omega)P_{\Omega}(\omega)$ and $P'_{\Gamma,\Omega}(\gamma,\omega)=P'_{\Gamma|\Omega}(\gamma,\omega)P'_{\Omega}(\omega)$, with 
%\begin{subequations}
%\label{eq:def_gamma_omega}
%\begin{align}
%P_{\Gamma|\Omega}(\gamma,\omega)&\coloneqq\Tr\left[\id\otimes E_{\lo{\gamma}}\frac{[\mathcal{M}_{\omega}\left(\hat{\rho}_{A|X}\right)](x_f)}{P_{\Omega}(\omega)}\id\otimes E_{\lo{\gamma}}^{\dagger}\right],\\
%P'_{\Gamma|\Omega}(\gamma,\omega)&\coloneqq\Tr\left[\id\otimes E_{\lo{\gamma}}\frac{[\mathcal{M}_{\omega}\left(\hat{\rho}'_{A|X}\right)](x_f)}{P'_{\Omega}(\omega)}\id\otimes E_{\lo{\gamma}}^{\dagger}\right].
%\end{align}
%\end{subequations}
%
\begin{subequations}
\label{eq:def_gamma_omega}
\begin{align}
P_{\Gamma|\Omega}(\gamma,\omega)&\coloneqq\Tr \left[\id\otimes E_{\gamma} \frac{\left[\mathcal{M}_{\omega}\left(\hat{\rho}_{A|X}\right)\right](x_f)}{P_{\Omega}(\omega)}\id\otimes E_{\gamma}^{ \dagger}\right]
=\Tr_B \left[\frac{E_{\gamma}\, \mathcal{E}_{\omega}(\rho_{B})\, E_{\gamma}^{\dagger}}{P_{\Omega}(\omega)}\right]\\ 
\intertext{and}
P'_{\Gamma|\Omega}(\gamma,\omega)&\coloneqq\Tr \left[\id\otimes E_{\gamma} \frac{\left[\mathcal{M}_{\omega}\left(\hat{\rho}'_{A|X}\right)\right](x_f)}{P'_{\Omega}(\omega)}\id\otimes E_{\gamma}^{ \dagger}\right]
=\Tr_B \left[\frac{E_{\gamma}\, \mathcal{E}_{\omega}(\rho'_{B})\, E_{\gamma}^{\dagger}}{P'_{\Omega}(\omega)}\right],
\end{align}
\end{subequations}
both of which are independent of $x_f$ and $a_f$.
%Since Bob's reduced state for the assemblage $\mathcal{M}_{\omega}\left(\hat{\rho}_{A|X}\right)$ is well-defined\lo{, both conditional distributions $P_{\Gamma|\Omega}$ and $P'_{\Gamma|\Omega}$ are} independent of $x_f$, in the same way that $P_{\Gamma}$ and $P'_{\Gamma}$ are independent of $x$, explicitly manifest in Eqs. \eqref{eq:Prob_gamma}. 
Now, since $X_f$ and $\Omega$ are independent variables, we can replace $P_{X_f|\Gamma}$ with $P_{X_f|\Gamma,\Omega}$ and exchange the order of the maximisation over $P_{X_f|\Gamma,\Omega}$ and the summation over $\omega$ in Eq. \eqref{eq:distancepormu0}. Furthermore, the optimal measurement operators for which the maximisation over  $\{E_{\gamma}\}$ is attained for each $\omega$ depend, of course, on $\omega$. Hence, we can also exchange the order of the summation over $\omega$ and the maximisation over the measurement operators if we make this dependence explicit by replacing, in Eqs. \eqref{eq:distancepormu0} and \eqref{eq:def_gamma_omega}, $\{E_{\gamma}\}$ with $\{E_{\gamma,\omega}\}$. With this, we write Eq. \eqref{eq:distancepormu0} as
\begin{align}
\nonumber
\sum_{\omega} P_{\Omega}(\omega) S_\text{A}\left( \frac{\mathcal{M}_{\omega}\left(\hat{\rho}_{A|X}\right)}{\Tr\left[\mathcal{M}_{\omega}(\rho_{A|X})\right]}\Bigg\|\frac{\mathcal{M}_{\omega}\left(\hat{\rho}'_{A|X}\right)}{\Tr\left[\mathcal{M}_{\omega}(\rho'_{A|X})\right]}\right)
= \max_{P_{X_f|\Gamma,\Omega},\{E_{\gamma,\omega}\}}\Bigg\{\sum_{\omega}P_{\Omega}(\omega)  \bigg[S_\text{C}\left(P_{\Gamma|\Omega}(\cdot, \omega)\|P'_{\Gamma|\Omega}(\cdot, \omega)\right)+&\\
\label{eq:distancepormu}
 \sum_{\gamma,x_f} P_{X_f,\Gamma|\Omega}(x_f,\gamma,\omega)\, S_\text{Q}\left( \frac{\id\otimes E_{\gamma,\omega}\left[\mathcal{M}_{\omega}\left(\hat{\rho}_{A|X}\right)\right](x_f)\id\otimes E_{\gamma,\omega}^{\dagger}}{P_{\Gamma,\Omega}(\gamma,\omega)}\:\bigg\|\:\frac{\id\otimes E_{\gamma,\omega} \left[\mathcal{M}_{\omega}\left(\hat{\rho}'_{A|X}\right)\right](x_f)\id\otimes E_{\gamma,\omega}^{\dagger}}{P'_{\Gamma,\Omega}(\gamma,\omega)} \right)\bigg]\Bigg\}.&
 \end{align}
%Eq. \eqref{eq:distancepormu} follows just by allowing $E_{\gamma}$ and $P_{X_f|\Gamma}$ depend explicitly on the value of $\omega$. This allows one to interchange the maximization and the summation over $\omega$. Hence, we employ the notation $E_{\gamma}^{m}$ and $P_{X_f|\Gamma,\Omega}$ to make this dependence explicit. In this way, $\{E^m_{\gamma}\}_{\gamma}$ is a set of Kraus operators for each $\omega$. 

Next, using Eqs. \eqref{eq:quant_rep}, \eqref{eq:def_W} and \eqref{eq:def_M_mu}, we write 
\begin{align}
\nonumber
S_\text{Q}\left( \frac{\id\otimes E_{\gamma,\omega}\left[\mathcal{M}_{\omega}\left(\hat{\rho}_{A|X}\right)\right](x_f)\id\otimes E_{\gamma,\omega}^{\dagger}}{P_{\Gamma,\Omega}(\gamma,\omega)}\:\bigg\|\:\frac{\id\otimes E_{\gamma,\omega} \left[\mathcal{M}_{\omega}\left(\hat{\rho}'_{A|X}\right)\right](x_f)\id\otimes E_{\gamma,\omega}^{ \dagger}}{P'_{\Gamma,\Omega}(\gamma,\omega)} \right)&=\\
\nonumber
S_\text{Q}\Bigg( \frac{\sum_{a_f,a,x} P_{X|X_f,\Omega}(x,x_f,\omega)P_{A_f|A,X,\Omega,X_f}(a_f,a,x,\omega,x_f)  
\: \ketbra{a_f}{a_f}\otimes E_{\gamma,\omega}\, K_{\omega}    \:  \varrho_{A|X}(a,x) \:  K^{\dagger}_{\omega}\, E_{\gamma,\omega}^{\dagger} }{P_{\Gamma,\Omega}(\gamma,\omega)}\:\bigg\|&\\
\nonumber
\frac{\sum_{a_f,a,x} P_{X|X_f,\Omega}(x,x_f,\omega)P_{A_f|A,X,\Omega,X_f}(a_f,a,x,\omega,x_f) 
\: \ketbra{a_f}{a_f}\otimes E_{\gamma,\omega}\, K_{\omega}    \:  \varrho'_{A|X}(a,x) \:  K^{\dagger}_{\omega}\, E_{\gamma,\omega}^{ \dagger}}{P'_{\Gamma,\Omega}(\gamma,\omega)} \Bigg)&\leq\\
\nonumber
\sum_{x}P_{X|X_f,\Omega}(x,x_f,\omega)S_\text{Q}\Bigg( \frac{\sum_{a_f,a} P_{A_f|A,X,\Omega,X_f}(a_f,a,x,\omega,x_f)  
\: \ketbra{a_f}{a_f}\otimes E_{\gamma,\omega}\, K_{\omega}    \:  \varrho_{A|X}(a,x) \:  K^{\dagger}_{\omega}\, E_{\gamma,\omega}^{\dagger}}{P_{\Gamma,\Omega}(\gamma,\omega)}\:\bigg\|&\\
\label{eq:before_sum_af}
\frac{\sum_{a_f,a} P_{A_f|A,X,\Omega,X_f}(a_f,a,x,\omega,x_f)
\: \ketbra{a_f}{a_f}\otimes E_{\gamma,\omega}\, K_{\omega}    \:  \varrho'_{A|X}(a,x) \:  K^{\dagger}_{\omega}\, E_{\gamma,\omega}^{ \dagger}}{P'_{\Gamma,\Omega}(\gamma,\omega)} \Bigg),
 \end{align}
where the inequality is due to Eq. \eqref{eq:convexity_Sq}. On the other hand, we note that there always exists a completely positive trace-preserving map $\mathcal{R}_{x,\omega,x_f}:\mathcal{L}(\mathcal{H}_E)\to\mathcal{L}({\mathcal{H}_E}_f)$ such that
\begin{equation}
\label{eq:auxiliary_map}
\mathcal{R}_{x,\omega,x_f}(\proj{a})=\sum_{a_f}\proj{a_f}P_{A_f|A,X,\Omega,X_f}(a_f,a,x,\omega,x_f).
\end{equation}
Hence, we can apply Eq. \eqref{eq:processingineq} to the von-Neumann relative entropy in the right-hand side of Eq. \eqref{eq:before_sum_af} with the map \eqref{eq:auxiliary_map}, to get
 \begin{align}
\nonumber
S_\text{Q}\Bigg( \frac{\sum_{a_f,a} P_{A_f|A,X,\Omega,X_f}(a_f,a,x,\omega,x_f)  
\: \ketbra{a_f}{a_f}\otimes E_{\gamma,\omega}\, K_{\omega}    \:  \varrho_{A|X}(a,x) \:  K^{\dagger}_{\omega}\, E_{\gamma,\omega}^{\dagger}}{P_{\Gamma,\Omega}(\gamma,\omega)}\:\bigg\|&\\
\nonumber
\frac{\sum_{a_f,a} P_{A_f|A,X,\Omega,X_f}(a_f,a,x,\omega,x_f)
\: \ketbra{a_f}{a_f}\otimes E_{\gamma,\omega}\, K_{\omega}    \:  \varrho'_{A|X}(a,x) \:  K^{\dagger}_{\omega}\, E_{\gamma,\omega}^{ \dagger}}{P'_{\Gamma,\Omega}(\gamma,\omega)} \Bigg)&\leq\\
\nonumber
S_\text{Q}\Bigg( \frac{\sum_{a} \: \ketbra{a}{a}\otimes E_{\gamma,\omega}\, K_{\omega}    \:  \varrho_{A|X}(a,x) \:  K^{\dagger}_{\omega}\, E_{\gamma,\omega}^{\dagger}}{P_{\Gamma,\Omega}(\gamma,\omega)}\:\bigg\|
\frac{\sum_{a} \: \ketbra{a}{a}\otimes E_{\gamma,\omega}\, K_{\omega}    \: \varrho'_{A|X}(a,x) \:  K^{\dagger}_{\omega}\, E_{\gamma,\omega}^{ \dagger}}{P'_{\Gamma,\Omega}(\gamma,\omega)} \Bigg)&=\\
\label{eq:use_map_R}
S_\text{Q}\Bigg( \frac{\id\otimes E_{\gamma,\omega}\, K_{\omega}    \:  \hat{\rho}_{A|X}(x) \:  K^{\dagger}_{\omega}\, E_{\gamma,\omega}^{\dagger}\otimes\id}{P_{\Gamma,\Omega}(\gamma,\omega)}\:\bigg\|
\frac{\id\otimes E_{\gamma,\omega}\, K_{\omega}    \: \hat{\rho}'_{A|X}(x) \:  K^{\dagger}_{\omega}\, E_{\gamma,\omega}^{ \dagger}\otimes\id}{P'_{\Gamma,\Omega}(\gamma,\omega)} \Bigg),
 \end{align}
where the quantum representation \eqref{eq:quant_rep} has been invoked again. Then, using Eqs. \eqref{eq:distancepormu}, \eqref{eq:before_sum_af}, and \eqref{eq:use_map_R}, we  obtain 
\begin{align}
\nonumber
\sum_{\omega} P_{\Omega}(\omega) S_\text{A}\left( \frac{\mathcal{M}_{\omega}\left(\hat{\rho}_{A|X}\right)}{\Tr\left[\mathcal{M}_{\omega}(\rho_{A|X})\right]}\Bigg\|\frac{\mathcal{M}_{\omega}\left(\hat{\rho}'_{A|X}\right)}{\Tr\left[\mathcal{M}_{\omega}(\rho'_{A|X})\right]}\right)
\leq \max_{P_{X_f|\Gamma,\Omega},\{E_{\gamma,\omega}\}}\Bigg\{\sum_{\omega}P_{\Omega}(\omega)  \bigg[S_\text{C}\left(P_{\Gamma|\Omega}(\cdot, \omega)\|P'_{\Gamma|\Omega}(\cdot, \omega)\right)&\\
\nonumber 
+ \sum_{\gamma,x_f,x}P_{X_f,\Gamma|\Omega}(x_f,\gamma,\omega)P_{X|X_f,\Omega}(x,x_f,\omega)\times&\\
\nonumber
 S_\text{Q}\Bigg( \frac{\id\otimes E_{\gamma,\omega}\, K_{\omega}    \:  \hat{\rho}_{A|X}(x) \:  K^{\dagger}_{\omega}\, E_{\gamma,\omega}^{\dagger}\otimes\id}{P_{\Gamma,\Omega}(\gamma,\omega)}\:\bigg\|
\frac{\id\otimes E_{\gamma,\omega}\, K_{\omega}    \: \hat{\rho}'_{A|X}(x) \:  K^{\dagger}_{\omega}\, E_{\gamma,\omega}^{ \dagger}\otimes\id}{P'_{\Gamma,\Omega}(\gamma,\omega)} \Bigg)\bigg]\Bigg\}&\\
\nonumber 
\leq \max_{P_{X_f|\Gamma,\Omega},\{E_{\gamma,\omega}\}}\Bigg\{\bigg[S_\text{C}\left(P_{\Gamma,\Omega}\|P'_{\Gamma,\Omega}\right)&\\
\nonumber 
+ \sum_{\gamma,x_f,x,\omega}P_{\Omega}(\omega)P_{X_f,\Gamma|\Omega}(x_f,\gamma,\omega)P_{X|X_f,\Omega,\Gamma}(x,x_f,\omega,\gamma)\times&\\
\label{eq:use_initial_relations}
 S_\text{Q}\Bigg( \frac{\id\otimes E_{\gamma,\omega}\, K_{\omega}    \:  \hat{\rho}_{A|X}(x) \:  K^{\dagger}_{\omega}\, E_{\gamma,\omega}^{\dagger}\otimes\id}{P_{\Gamma,\Omega}(\gamma,\omega)}\:\bigg\|
\frac{\id\otimes E_{\gamma,\omega}\, K_{\omega}    \: \hat{\rho}'_{A|X}(x) \:  K^{\dagger}_{\omega}\, E_{\gamma,\omega}^{ \dagger}\otimes\id}{P'_{\Gamma,\Omega}(\gamma,\omega)} \Bigg)\bigg]\Bigg\}&,
 \end{align}
where the inequality \eqref{eq:use_initial_relations} follows from Eq. \eqref{eq:knullback} and from replacing $P_{X|X_f,\Omega}(x,x_f,\omega)$ with $P_{X|X_f,\Omega,\Gamma}(x,x_f,\omega,\gamma)$, which cannot decrease the value of the resulting maximum. 

Finally, using that, due to Bayes' theorem, it holds that 
\begin{equation}
\sum_{x_f}P_{\Omega}(\omega)P_{X_f,\Gamma|\Omega}(x_f,\gamma,\omega)P_{X|X_f,\Omega,\Gamma}(x,x_f,\omega,\gamma)=P_{X,\Gamma,\Omega}(x,\gamma,\omega),
\end{equation}
and introducing the joint variable $\Xi\coloneqq(\Gamma,\Omega)$, with values $\xi\coloneqq(\gamma,\omega)$, and the joint Kraus operators $T_{\xi}\coloneqq E_{\gamma,\omega}K_{\omega}$, which satisfy the normalisation condition $\sum_{\xi}T^{\dagger}_{\xi}T_{\xi}=\sum_{\gamma,\omega}E^{\dagger}_{\gamma,\omega}K^{\dagger}_{\omega}E_{\gamma,\omega}K_{\omega}=\id$, we write the inequality \eqref{eq:use_initial_relations} as
\begin{align}
\nonumber
\sum_{\omega} P_{\Omega}(\omega) S_\text{A}\left( \frac{\mathcal{M}_{\omega}\left(\hat{\rho}_{A|X}\right)}{\Tr\left[\mathcal{M}_{\omega}(\rho_{A|X})\right]}\Bigg\|\frac{\mathcal{M}_{\omega}\left(\hat{\rho}'_{A|X}\right)}{\Tr\left[\mathcal{M}_{\omega}(\rho'_{A|X})\right]}\right)
\leq \\ 
\label{eq:last_proof_Lemma_1}
\max_{P_{X_f|\Xi},\{T_{\xi}\}}\Bigg\{\bigg[S_\text{C}(P_{\Xi}\|P'_{\Xi})+ \sum_{x,\xi}P_{X,\Xi}(x,\xi)
\times S_\text{Q}\Bigg( \frac{\id\otimes T_{\xi}   \:  \hat{\rho}_{A|X}(x) \:  T_{\xi}^{\dagger}\otimes\id}{P_{\Xi}(\xi)}\:\bigg\|
\frac{\id\otimes T_{\xi}  \: \hat{\rho}'_{A|X}(x) \:  T_{\xi}^{ \dagger}\otimes\id}{P'_{\Xi}(\xi)} \Bigg)\bigg]\Bigg\}.
 \end{align}
By Definition \ref{def:rel_ent_assemb}, the right-hand side of Eq. \eqref{eq:last_proof_Lemma_1} coincides with the right-hand side of Eq. \eqref{eq:strongmonotonicitydist}.
\end{proof}
%%%%%%%%%%%%%%%%%%%%%%%%%%%%%%%%%%%%%%%%%%%%%%%%%%%%%%%%%%%%%
\subsection{Proof of Lemma \ref{lem:convex}}
For the proof of this lemma, it is useful to re-express Eq. \eqref{eqmat:relativeentropy} in terms of abstract flag states  representing the outcomes of Bob's generalized quantum measurements. Introducing an auxiliary extension Hilbert space $\mathcal{H}_{E_B}$ and an orthonormal basis of it $\{\ket{\gamma}\}$, where each basis member encodes the value $\gamma$ of the measurement outcomes, and using that $\sum_{x}P_{X|\Gamma}(x,\gamma)=1$ for all $\gamma$, we write 
\begin{align}
\nonumber 
S_\text{A}(\rho_{A|X}\|\rho'_{A|X})&%=\max_{P_{X|\Gamma},\{E_{\gamma}\}} \bigg\{S_C(P_{\Gamma}\|P'_{\Gamma})+\sum_{\gamma,x}P_{X,\Gamma}(x,\gamma)S_Q\bigg( \frac\id \otimes}E_{\gamma}\hat{\rho}_{A|X}(x)\id \otimes E_{\gamma}^{\dagger}}{P_{\Gamma}(\gamma)}\bigg\|\frac{\id \otimes E_{\gamma}\hat{\rho}'_{A|X}(x)\id \otimes E_{\gamma}^{\dagger}}{P'_{\Gamma}(\gamma)} \bigg)\bigg\}\\
%\nonumber
=\max_{P_{X|\Gamma},\{E_{\gamma}\},}  \Bigg\{\sum_{x}P_{X|\Gamma}(x,\gamma) \Bigg[S_\text{C}(P_{\Gamma}\|P'_{\Gamma})\\
\nonumber
&+ \sum_{\gamma}P_{\Gamma}(\gamma)S_Q\left( \frac{\id \otimes E_{\gamma}\hat{\rho}_{A|X}(x)\id \otimes E_{\gamma}^{\dagger}}{P_{\Gamma}(\gamma)}\bigg\|\frac{\id \otimes E_{\gamma}\hat{\rho}'_{A|X}(x)\id \otimes E_{\gamma}^{\dagger}}{P'_{\Gamma}(\gamma)} \right)\Bigg]\Bigg\}\\
\label{eq:quant_rep_gamma}
&=\max_{P_{X|\Gamma}, \{E_{\gamma}\}}  \left[\sum_{x}P_{X|\Gamma}(x,\gamma) S_Q\left(\sum_{\gamma}\proj{\gamma} \otimes E_{\gamma}\hat{\rho}_{A|X}(x)\id \otimes E_{\gamma}^{\dagger} \Big\| \sum_{\gamma}\proj{\gamma} \otimes E_{\gamma}\hat{\rho}'_{A|X}(x)\id \otimes E_{\gamma}^{\dagger} \right)\right]
\end{align}

We can now prove the lemma.
\begin{proof}[Proof of Lemma \ref{lem:convex}]
Using Eq. \eqref{eq:quant_rep_gamma}, we write the left-hand side of Eq. \eqref{eq:convexity} as
\begin{align}
\nonumber 
S_\text{A}\left(\sum_{j}\mu^{(j)} \rho^{(j)}_{A|X} \Big\|\sum_{j}\mu^{(j)} \rho'^{(j)}_{A|X}\right)\\
\nonumber 
=\max_{P_{X|\Gamma},\{E_{\gamma}\}} \left[\sum_{x}P_{X|\Gamma}(x,\gamma)\,
S_Q\left(\sum_{j}\mu^{(j)} \sum_{\gamma}\proj{\gamma} \otimes E_{\gamma}\hat{\rho}^{(j)}_{A|X}(x)\id \otimes E_{\gamma}^{\dagger} \Big\| \sum_{j}\mu^{(j)} \sum_{\gamma}\proj{\gamma} \otimes E_{\gamma}\hat{\rho}'^{(j)}_{A|X}(x)\id \otimes E_{\gamma}^{\dagger} \right)\right]\\
\label{eq:use_joint_convexity}
\leq\max_{P_{X|\Gamma},\{E_{\gamma}\}}  \left[\sum_{x,j}\mu^{(j)} P_{X|\Gamma}(x,\gamma)\, S_Q\left(\sum_{\gamma}\proj{\gamma} \otimes E_{\gamma}\hat{\rho}^{(j)}_{A|X}(x)\id \otimes E_{\gamma}^{\dagger} \Big\| \sum_{\gamma}\proj{\gamma} \otimes E_{\gamma}\hat{\rho}'^{(j)}_{A|X}(x)\id \otimes E_{\gamma}^{\dagger} \right)\right]\\
\label{eq:exchange_order}
\leq\sum_{j}\mu^{(j)}\max_{\{E_{\gamma,j}\},P_{X|\Gamma, J}}  \left[\sum_{x} P_{X|\Gamma, J}(x,\gamma, j)
S_Q\left(\sum_{\gamma}\proj{\gamma} \otimes E_{\gamma,j}\hat{\rho}^{(j)}_{A|X}(x)\id \otimes E_{\gamma,j}^{\dagger} \Big\| \sum_{\gamma}\proj{\gamma} \otimes E_{\gamma,j}\hat{\rho}'^{(j)}_{A|X}(x)\id \otimes E_{\gamma,j}^{\dagger} \right)\right],
\end{align}
where \eqref{eq:use_joint_convexity} follows from Eq. \eqref{eq:convexity_Sq} and, in Eq. \eqref{eq:exchange_order}, we exchanged the order of the maximization and the summation over $j$ by respectively replacing $\{E_{\gamma}\}$ and $P_{X|\Gamma}$ with $\{E_{\gamma,j}\}$ and  $P_{X|\Gamma, J}$, of elements $P_{X|\Gamma, J}(x,\gamma, j)$. Using again Eq. \eqref{eq:quant_rep_gamma}, one sees that, by Definition \ref{def:rel_ent_assemb}, the right-hand side of Eq. \eqref{eq:exchange_order} coincides with the right-hand side of Eq. \eqref{eq:convexity}.
\end{proof}
\end{widetext}


\begin{thebibliography}{99}%

\bibitem{Schrodinger35}
E. Schr\"{o}dinger, {\it Discussion of probability relations between separated systems}, Proc. Camb. Phil. Soc. {\bf 31}, 555 (1935).

\bibitem{Wiseman07}
H. M. Wiseman, S. J. Jones and A. C. Doherty, {\it Steering, Entanglement, Nonlocality, and the Einstein-Podolsky-Rosen Paradox}, Phys. Rev. Lett. {\bf 98}, 140402 (2007); S. J. Jones {\it et al}., Phys. Rev. A
{\bf 76}, 052116 (2007).

\bibitem{Reid09}
M. D. Reid, P. D. Drummond, W. P. Bowen, E. G. Cavalcanti, P. K. Lam, H. A. Bachor, U. L. Andersen, and G. Leuchs, {\it \emph{Colloquium:} The Einstein-Podolsky-Rosen paradox: From concepts to applications}, 
Rev. Mod. Phys. {\bf 81}, 1727 (2009).

\bibitem{Horodecki09}
R. Horodecki, P. Horodecki, M. Horodecki, and K.
Horodecki, {\it Quantum entanglement}, Rev. Mod. Phys. {\bf 81}, 865 (2009).

\bibitem{Brunner13}
N. Brunner, D. Cavalcanti, S. Pironio, V. Scarani and S.
Wehner, {\it Bell nonlocality}, Rev. Mod. Phys. {\bf 86}, 419 (2014).

\bibitem{Cavalcanti09}
E. G. Cavalcanti, S. J. Jones, H. M. Wiseman and M. D.
Reid, {\it Experimental criteria for steering and the Einstein-Podolsky-Rosen paradox}, Phys. Rev. A {\bf 80}, 032112 (2009).

\bibitem{Experiments}
Z. Y. Ou, S. F. Pereira, H. J. Kimble, and K. C. Peng, {\it Realization of the Einstein-Podolsky-Rosen paradox for continuous variables}, 
Phys. Rev. Lett. {\bf 68}, 3663 (1992); W. P. Bowen, R. Schnabel, and P. K. Lam, {\it Experimental Investigation of Criteria for Continuous Variable Entanglement},  Phys. Rev. Lett. {\bf 90}, 043601 (2003); D.-H. Smith {et al}., {\it Conclusive quantum steering with superconducting transition-edge sensors},  Nat. Commun. {\bf 3}, 625 (2012); A. J. Bennet et al.,{\it Arbitrarily Loss-Tolerant Einstein-Podolsky-Rosen Steering Allowing a Demonstration over 1 km of Optical Fiber with No Detection Loophole}, Phys. Rev. X {\bf 2}, 031003 (2012); V. H\"andchen {\it et al}., {\it Observation of one-way Einstein-Podolsky-Rosen steering},  Nat. Phot. {\bf 6}, 598 (2012); S. Steinlechner, J. Bauchrowitz, T. Eberle, and R. Schnabel, {\it Strong Einstein-Podolsky-Rosen steering with unconditional entangled states}, Phys. Rev. A {\bf 87}, 022104 (2013).

\bibitem{Saunders10}
D. J. Saunders, S. J. Jones, H. M. Wiseman, and G. J. Pryde, {\it Experimental EPR-steering using Bell-local states}, Nat. Phys. {\bf 6}, 845 (2010).

\bibitem{loopholefree}
B. Wittmann, S. Ramelow, F. Steinlechner, N. K. Langford, N. Brunner, H. Wiseman, R. Ursin, A. Zeilinger, {\it Loophole-free Einstein-Podolsky-Rosen experiment via quantum steering}, New J. Phys. {\bf 14}, 053030 (2012).

\bibitem{Branciard12}
C. Branciard, E. G. Cavalcanti, S. P. Walborn, V. Scarani and H. M. Wiseman, {\it One-sided device-independent quantum key distribution: Security, feasibility, and the connection with steering}, Phys. Rev. A {\bf 85}, 010301(R) (2012).

\bibitem{He13}
Q. Y. He and M. D. Reid, {\it Genuine Multipartite Einstein-Podolsky-Rosen Steering}, Phys. Rev. Lett. {\bf 111}, 250403 (2013).

\bibitem{Deviceindependent}
J. Barrett, L. Hardy,  and A. Kent, {\it No Signaling and Quantum Key Distribution}, Phys. Rev. Lett. {\bf 95}, 010503 (2005); A. Ac\'in, N. Gisin, and L. Masanes, {\it From Bell’s Theorem to Secure Quantum Key Distribution}, Phys. Rev. Lett. {\bf 97}, 120405 (2006); A. Ac\'in {\it et al}., {\it Device-Independent Security of Quantum Cryptography against Collective Attacks}, Phys. Rev. Lett. {\bf 98}, 230501 (2007).

\bibitem{VedralPlenio}
V. Vedral and M. B. Plenio, {\it Entanglement measures and purification procedures}, Phys. Rev. A {\bf 57}, 1619 (1998); M. B. Plenio and S. Virmani, {\it An introduction to entanglement measures}, Quant. Inf. Comput. {\bf 7}, 1 (2007).

\bibitem{Brandao08}
F. G. S. L. Brand\~{a}o and M. B. Plenio, {\it Entanglement theory and the second law of thermodynamics}, Nature Phys. {\bf 4}, 873 (2008); F. G. S. L. Brand\~{a}o and M. B. Plenio, {\it A Reversible Theory of Entanglement and its Relation to the Second Law},  Comm. Math. Phys. 295, 829 (2010).

\bibitem{Bennett96}
C. H. Bennett, D. P. DiVincenzo, J. Smolin, and W. K.
Wootters, {\it Mixed-state entanglement and quantum error correction}, Phys. Rev. A {\bf 54}, 3824 (1996).

\bibitem{Vedral97}
V. Vedral, M. B. Plenio, M. A. Rippin and P. L. Knight, {\it Quantifying Entanglement}, Phys. Rev. Lett {\bf 78,} 2275 (1997)

\bibitem{Jonathan99}
D. Jonathan and M. B. Plenio, {\it Entanglement-Assisted Local Manipulation of Pure Quantum States}, Phys. Rev. Lett. {\bf 83,} 3566 (1999)

\bibitem{Brandao13}
F. G. S. L. Brand\~{a}o, M. Horodecki, J. Oppenheim, J. M. Renes, and R. W. Spekkens, {\it Resource Theory of Quantum States Out of Thermal Equilibrium}, Phys. Rev. Lett. {\bf 111}, 250404 (2013).


\bibitem{Ahmadi13}
M. Ahmadi, D. Jennings and T. Rudolph, {\it The Wigner-Araki-Yanase theorem and the quantum resource theory of asymmetry}, New J. Phys. {\bf 15}, 013057 (2013).

\bibitem{Gour08}
G. Gour and R. W. Spekkens, {\it The resource theory of quantum reference frames: manipulations and monotones}, New J. Phys. {\bf 10}, 033023 (2008).

\bibitem{Levi14}
F. Levi and F. Mintert, {\it A quantitative theory of coherent delocalization}, New J. Phys. {\bf 16}, 033007 (2014).

\bibitem{Baumgratz13}
T. Baumgratz, M. Cramer, and M. B. Plenio, {\it Quantifying Coherence}, Phys. Rev. Lett. {\bf 113}, 140401 (2014).

\bibitem{Pusey13}
M. F. Pusey, {\it Negativity and steering: A stronger Peres conjecture}, Phys. Rev. A {\bf 88}, 032313 (2013). 

\bibitem{Skrzypczyk13}
P. Skrzypczyk, M. Navascu\'es, and D. Cavalcanti, {\it Quantifying Einstein-Podolsky-Rosen Steering}, Phys. Rev. Lett. {\bf 112}, 180404 (2014).

\bibitem{Piani15}
M. Piani and J. Watrous, {\it Necessary and Sufficient Quantum Information Characterization of Einstein-Podolsky-Rosen Steering}, Phys. Rev. Lett. {\bf 114}, 060404 (2015).

\bibitem{Bowles14}
J. Bowles, T. V\'ertesi, M. T. Quintino and N. Brunner. {\it One-way Einstein-Podolsky-Rosen Steering}, Phys. Rev. Lett. {\bf 112,} 200402 (2014).

\bibitem{Quintino14}
M. T. Quintino, T. V\'ertesi, N. Brunner. {\it Joint Measurability, Einstein-Podolsky-Rosen Steering, and Bell Nonlocality}, Phys. Rev. Lett. {\bf113,} 160402 (2014).

\bibitem{Belen15}
A. B. Sainz, N. Brunner, D. Cavalcanti, P. Skrzypczyk, and T. V\'ertesi, {\it Post-quantum steering}, arXiv:1505.01430 (2015).

\bibitem{footnote2}
One could in principle still attempt a resource theory of steerable quantum states as defined in \cite{Wiseman07, Reid09}. However, since the mathematical condition that defines a quantum state as steerable involves an optimisation over its potential measurements, it is unclear what the precise resource to account for is.  An analogous difficulty is found with Bell non-locality as a resource \cite{Gallego12,deVicente14}. There, one could  conceive a resource theory of ``non-local quantum states'', defined as those able to yield, under local measurements, a non-local probability distribution. However, for the same reasons, non-locality defined directly in terms of probability distributions plays a more relevant role.

\bibitem{Oreshkov09}
O. Oreshkov and J. Calsamiglia, {\it Distinguishability measures between ensembles of quantum states}, Phys. Rev. A {\bf 79}, 032336 (2009).

\bibitem{Ecker91}
A. K. Ekert, {\it Quantum cryptography based on Bell’s theorem}, Phys. Rev. Lett. {\bf 67}, 661(1991).

\bibitem{Gallego12}
R. Gallego, L. E. W\"urflinger, A. Ac\'in, and M. Navascu\'es, {\it Operational Framework for Nonlocality}, Phys. Rev. Lett. {\bf 109}, 070401  (2012).

\bibitem{deVicente14}
J. I. de Vicente, {\it On nonlocality as a resource theory and nonlocality measures}, J. Phys. A: Math. Theor. {\bf 47}, 424017 (2014).

\bibitem{Navascues14}
B. Lang, T. V\'ertesi and M. Navascu\'es, {\it Closed sets of correlations: answers from the zoo}, arXiv:1402.2850 (2014).


\bibitem{footnote1}
The terminology ``deterministic maps'' refers throughout to probability (trace) preserving classical (quantum) maps%. %These are maps such that, given an input bit (state), \lo{they} generate an output bit (state), respectively, with certainty
, i.e. those that never cause an abortion. %does not mean that the output cannot be chosen at random. That is, this 
For classical maps, for instance, this should not be confused with maps where the output bit is a Kronecker delta function of the input bit. In turn, the term ``stochastic" is used throughout to refer to non probability-preserving classical maps or non trace-preserving quantum transformations, which do not occur with certainty.

\bibitem{footnote0}
A natural question (which we leave open) is whether there exists a definition of relative entropy between assemblages that is non-increasing under generic assemblage transformations instead of just 1W-LOCCs, so that it can be understood as measure of distinguishability under fully general strategies. For quantum states, that is the case of $S_\text{Q}$, for instance, which is non-increasing under not only LOCCs but also under any completely positive map. 
However, note that 1W-LOCC-monotonicity of $S_\text{A}$ suffices to introduce a steering monotone.

\bibitem{Kogias14}
I. Kogias, A. R. Lee, S. Ragy, and G. Adesso, {\it Quantification of Gaussian Quantum Steering}, Phys. Rev. Lett. {\bf 114}, 060403 (2015).

\bibitem{Vidal99}
 G. Vidal, {\it Entanglement of Pure States for a Single Copy}, Phys. Rev. Lett. {\bf 83}, 1046 (1999).

\bibitem{Duer99}
W. D\"{u}r, G. Vidal, and J. I. Cirac, {\it Three qubits can be entangled in two inequivalent ways}, Phys. Rev. A {\bf 62}, 062314 (2000).

\bibitem{Owari04}
M. Owari, K. Matsumoto, and M. Murao, {\it Entanglement convertibility for infinite-dimensional pure bipartite states}, Phys. Rev. A {\bf 70}, 050301(R) (2004).

 \bibitem{Nielsen99}
M. A. Nielsen, {\it Conditions for a Class of Entanglement Transformations}, Phys. Rev. Lett. {\bf 83},  436 (1999). 

\end{thebibliography}

\begin{thebibliography}{99}%
\bibitem[42]{Vedral98}
V. Vedral and M. B. Plenio, {\it Entanglement measures and purification procedures}, Phys. Rev. A {\bf 57}, 1619 (1998).




\bibitem[43]{VidalTarrach99}
G. Vidal and R. Tarrach,  {\it Robustness of entanglement}, Phys. Rev. A {\bf 59}, 141 (1999).

\bibitem[44]{Lieb73}
E. H. Lieb, {\it Convex Trace Functions and the Wigner-Yanase-Dyson Conjecture}, Adv. Math. {\bf 11}, 267 (1973).
\end{thebibliography}
\end{document}